\documentclass{article}

\usepackage[square,sort,comma,numbers]{natbib}

\usepackage[final]{neurips_2022}

\usepackage{microtype}
\usepackage{graphicx}
\usepackage{subfigure}
\usepackage{booktabs}
\usepackage[utf8]{inputenc} 
\usepackage[T1]{fontenc}    
\usepackage{hyperref}       
\usepackage{url}            
\usepackage{booktabs}       
\usepackage{amsfonts}       
\usepackage{nicefrac}       
\usepackage{microtype}      
\usepackage{xcolor}         
\usepackage{amsmath}
\usepackage{amssymb}
\usepackage{mathtools}
\usepackage{amsthm}
\usepackage[ruled]{algorithm2e}
\newtheorem{theorem}{Theorem}[section]
\newtheorem{proposition}[theorem]{Proposition}
\newtheorem{lemma}[theorem]{Lemma}
\newtheorem{corollary}[theorem]{Corollary}

\theoremstyle{definition}
\newtheorem{definition}[theorem]{Definition}

\newtheorem{remark}[theorem]{Remark}

\usepackage{multirow}
\usepackage{multicol}
\usepackage{booktabs}
\usepackage{makecell}

\title{Benefits of Permutation-Equivariance in\\Auction Mechanisms}

\author{%
  Tian Qin\textsuperscript{1, 2}\thanks{Both authors contributed equally.}
  \quad
  Fengxiang He\textsuperscript{2}\footnotemark[1]
  \quad
  Dingfeng Shi\textsuperscript{2}
  \quad
  Wenbing Huang\textsuperscript{3, 4}
  \quad
  Dacheng Tao\textsuperscript{2}\\
  $^1$University of Science and Technology of China\quad
  $^2$JD Explore Academy, JD.com Inc.\\
  $^3$Gaoling School of Artificial Intelligence, Renmin University of China\\
  $^4$Beijing Key Laboratory of Big Data Management and Analysis Methods\\
  \texttt{\nolinkurl{tqin@mail.ustc.edu.cn}, \nolinkurl{fengxiang.f.he@gmail.com},}\\ \texttt{\nolinkurl{shidingfeng@buaa.edu.cn}, \nolinkurl{hwenbing@126.com}, \nolinkurl{dacheng.tao@gmail.com}}
 }

\begin{document}

\maketitle

\begin{abstract}
Designing an incentive-compatible auction mechanism that maximizes the auctioneer's revenue while minimizes the bidders’ ex-post regret is an important yet intricate problem in economics. Remarkable progress has been achieved through learning the optimal auction mechanism by neural networks. In this paper, we consider the popular {\it additive valuation} and {\it symmetric valuation} setting; {\it i.e.}, the valuation for a set of items is defined as the sum of all items’ valuations in the set, and the valuation distribution is invariant when the bidders and/or the items are permutated. We prove that permutation-equivariant neural networks have significant advantages: the permutation-equivariance decreases the expected ex-post regret, improves the model generalizability, while maintains the expected revenue invariant. This implies that the permutation-equivariance helps approach the theoretically optimal {\it dominant strategy incentive compatible} condition, and reduces the required sample complexity for desired generalization. Extensive experiments fully support our theory. To our best knowledge, this is the first work towards understanding the benefits of permutation-equivariance in auction mechanisms. 
\end{abstract}
\section{Introduction}\label{intro}
Optimal auction design \cite{2016Twenty} has wide applications in economics, including computational advertising \citep{jansen2008sponsored}, resource allocation \citep{huang2008auction}, and supply chain \citep{2019Auction}.  
In an auction, every 
bidder has a private valuation profile 
over all items, and accordingly, submits her bid profile. An auctioneer collects the bids from all bidders, and determines a feasible item allocation 
to the bidders as well as the prices 
the bidders need to pay. Consequently, every bidder receives her utility. 
From the auctioneer's perspective, the optimal auction {mechanism} is required to maximize her {\it revenue}, defined as the sum of all bidders' payments. 
From the aspect of the bidders, the optimal auction {mechanism} needs to incentivize every bidder to bid their truthful valuation profiles (truthful bidding). 
This is summarized as the {\it dominant strategy incentive compatible} (DSIC) condition; {\it i.e.}, truthful bidding is always the dominant strategy for every bidder \cite{2007Algorithmic}.

The optimal auction mechanism can be approximated via neural networks \cite{dutting2019optimal, rahme2021permutation, duan2022context}. The ``approximation error'', or the ``distance'' to the DSIC condition, is usually measured by the {\it ex-post regret}, defined as the gap between the bidder's utility of truthful bidding and the utility when her bid profile is only the best to herself (selfish bidding), while the bid profiles of all other bidders are fixed in both cases \cite{dutting2019optimal}.
When a bidder's ex-post regret is $0$, truthful bidding is her dominant strategy. Therefore, the optimal auction design can be modeled as a linear programming problem, where the object is to maximize the expected revenue subject to the expected ex-post regret being $0$ for all bidders \cite{dutting2019optimal}. Another major consideration in learning the optimal mechanism is the generalizability to unseen data, usually measured by the {\it generalization bound}, {\it i.e.}, the upper bound of the gap between the expected revenue/ex-post regret and their empirical counterparts on the training data  \cite{dutting2019optimal}. 

In this paper, we consider the popular setting of {\it additive valuation} and {\it symmetric valuation} \cite{dutting2019optimal,rahme2021permutation,duan2022context}. The additive valuation condition defines the valuation for a set of items as the sum of the valuations for all items in this set. The symmetric valuation condition assumes the joint distribution of all bidders' valuation profiles to be invariant when bidders and/or items are permutated. This setting covers many applications in practice. For example, when items are independent, the additive valuation condition holds. Moreover, if the auction is anonymous or the order of the items is not prior-known, the symmetric valuation condition holds.



We demonstrate that permutation-equivariant models have significant advantages in learning the optimal auction mechanism as follows. \textbf{(1)} We prove that the permutation-equivariance in auction mechanisms decreases the expected ex-post regret while maintaining the expected revenue invariant. Conversely, and equivalently, the permutation-equivariance promises a larger expected revenue, when the expected ex-post regret is fixed. 
\textbf{(2)} We show that the permutation-equivariance of auction mechanisms reduces the required sample complexity for desirable generalizability. We prove that the $l_{\infty,1}$-distance between any two mechanisms in the mechanism space decreases when they are projected to the permutation-equivariant mechanism (sub-)space. This smaller distance implies a smaller covering number of the permutation-equivariant mechanism space, which further leads to a small generalization bound \cite{dutting2019optimal}.

We further provide an explanation for the learning process of non-permutation-equivariant neural networks (NPE-NNs). 
In learning the optimal auction mechanism by an NPE-NN, we show that an extra positive term exists in the quadratic penalty of the ex-post regret based on the result \textbf{(1)}. This term serves as a regularizer to penalize the ``non-permutation-equivariance''. Moreover, this regularizer also interferes the revenue maximization, and thus affects the learning performance of NPE-NNs. This further explains the advantages of permutation-equivariance in auction design.

Experiments in extensive auction settings are conducted to verify our theory. We design permutation-equivariant versions of RegretNet (RegretNet-PE and RegretNet-test) by projecting the RegretNet \cite{dutting2019optimal} to the permutation-equivariant mechanism space in the training and test stage respectively. 
The empirical results show that permutation-equivariance helps: {(1) significantly improve the revenue while maintain the same ex-post regret;
(2) record the same revenue with a significantly lower ex-post regret;
and (3) narrow the generalization gaps between the training ex-post regret and its test counterpart.} These results fully support our theory.

\textbf{Related works.}\label{rela} 
Myerson completely solves the optimal auction design problem in one-item auctions \cite{myerson1981optimal}. However, solutions are not clear when the number of bidders/items exceeds one \cite{dutting2019optimal}. 
{ 
Initial attempts have been presented on the characterization of optimal auction mechanisms \cite{manelli2006bundling,pavlov2011optimal,daskalakis2017strong} and algorithmic solutions \cite{2011Bayesian,2013The,2012An}.} Remarkable advances have been made in the required sample complexity for learning the optimal auction mechanism in various settings, including single-item auctions \cite{cole2014sample,2013Learning,huang2008auction}, multi-item single-bidder auctions \cite{2014Sampling}, combinatorial auctions \cite{balcan2016sample,2017A}, and allocation mechanisms \cite{2016A}.
Machine learning-based auction design (automated auction design) have obtained considerable progress \cite{conitzer2002complexity,conitzer2004self,sandholm2015automated}. The optimal auction design is modeled as a linear programming problem \cite{conitzer2002complexity, conitzer2004self}. However, early works suffer from the scalability issues that the number of the constraints grows exponentially when the bidder number and the item numbers increase. To address this issue, recent works propose to learn the optimal auction mechanism by deep learning. RegretNet 
is designed for multi-bidder and multi-item settings \cite{dutting2019optimal}. Then, RegretNet is developed to meet more restrictive constraints, such as the budget condition \cite{feng2018deep} and the certifying strategyproof condition \cite{curry2020certifying}. 
Rahme {\it et al.} \cite{rahme2021permutation} propose the first equivariant neural network-based auction mechanism
design method with significant empirical advantages. 
Ivanov {\it et al.} \cite{ivanov2022optimal} propose a RegretFormer which (1) introduces attention layers to RegretNet to learn permutation-equivariant auction mechanisms, and (2) adopts a new interpretable loss function to control the revenue-regret trade-off.
Duan {\it et al.} \cite{duan2022context} extend the applicable domain to contextual settings. 
All these works make remarkable contributions in designing new algorithms from the empirical aspect only. 
However, the theoretical foundations are still elusive. 
To our best knowledge, our paper is the first work on theoretically studying the benefits of permutation equivariance in auction design via deep learning.


\section{Notations and Preliminaries}\label{pre}
\textbf{Auction.}
Suppose $n$ bidders are bidding $m$ items in an auction. Every bidder $i$ has her bidder-context (feature) $x_i\in\mathcal{X}$, while every item $j$ is associated with its item-context (feature) $y_j\in\mathcal{Y}$. The bidder $i$ has a private valuation $v_{i j}\in\mathcal{V}\subset\mathbb{R}_{\ge 0}$ for the item $j$, which is sampled from a conditioned distribution $\mathbb{P}(\cdot|x_i,y_j)$. The value profile $v_i = (v_{i 1},\ldots,v_{i m})$ is unknown to the auctioneer. For the simplicity, we define $x= (x_1^T,\ldots,x_n^T)^T$, $y=(y_1,\ldots,y_m)$, $v=(v_1^T,\ldots,v_n^T)^T$, $v_{-i}=(v_1^T,\ldots,v_{i-1}^T,v_{i+1}^T,\ldots,v_n^T)^T$, and $(v_i',v_{-i})=(v_1^T,\ldots,(v_{i}')^T\ldots,v_n^T)^T$.

Every bidder submits a bid profile $b_i$ to the auctioneer according to her valuation profile. Then, the auctioneer determines a feasible item allocation $g(b,x,y)$ and corresponding payments $p(b,x,y)$ as per an auction mechanism $(g,p)$. Consequently, every bidder receives her utility as
\begin{equation*}
    u_i(v_i,b,x,y)=\sum_{j=1}^m g_{i j}(b,x,y)\cdot v_{i j}-p_i(b,x,y).
\end{equation*}
The auction mechanism $(g,p)$ consists of an allocation rule $g:\mathbb{R}^{n\times m}\times\mathcal{X}^n\times\mathcal{Y}^m\to\mathbb{R}^{n\times m}$ and a payment rule $p:\mathbb{R}^{n\times m}\times\mathcal{X}^n\times\mathcal{Y}^m\to\mathbb{R}^{n\times m}$, where $g_{i j}$ is the probability of allocating item $j$ to the bidder $i$, and $p_i=\sum_{j=1}^m p_{ij}$ is the price that the bidder $i$ should pay. To avoid allocating an item over once, the allocation rule is constrained such that $\sum_{i=1}^n g_{i j}(b,x,y)\le 1$ for all $j\in[m]$.  Every $v_i$ in our notations can be replaced by $b_i$.  
Thus, we can define the similar notations: $b_{-i}=(b_1^T,\ldots,b_{i-1}^T,b_{i+1}^T,\ldots,b_n^T)^T$, $(b_i,v_{-i})=(v_1^T,\ldots,(b_{i})^T\ldots,v_n^T)^T$ and $(v_i,b_{-i})=(b_1^T,\ldots,(v_{i})^T\ldots,b_n^T)^T$.

\textbf{Optimal auction mechanism.}
An auction mechanism $(g,p)$ is defined to be {\it dominant strategy incentive compatible} (DSIC), if truthful bidding is always a dominant strategy of every bidder; {\it i.e.},
\begin{equation*}
    u_i(v_i,(v_i,b_{-i}),x,y)\ge u_i(v_i,b,x,y),
\end{equation*}
for all $i\in [n]$, $v,b\in\mathcal{V}^{n\times m}$, $x\in\mathcal{X}^n$ and $y\in\mathcal{Y}^m$.    
In addition, an auction mechanism $(g,p)$ is called {\it individually rational} (IR), if for any bidder-contexts $x\in\mathcal{X}^n$, any item-contexts $y\in\mathcal{Y}^m$, any bidder $i\in [n]$ $x\in\mathcal{X}^n$, valuation profile and bid profile $v,b\in\mathcal{V}^{n\times m}$,  truthful bidding always leads to a non-negative utility, {\it i.e.},
\begin{equation*}
    u_i(v_i,(v_i,b_{-i}),x,y)\ge 0.
\end{equation*}
If an auction mechanism is DSIC and IR, a rational bidder with an obvious dominant strategy will play it (bidding truthfully).
Moreover, an optimal auction mechanism is required to maximize the auctioneer's expected revenue $rev=\mathbb{E}_{(v,x,y)}\big[\sum_{i=1}^n p_i(v,x,y)\big]$.

\textbf{Auction design.} The ex-post regret $reg_{i}(v,x,y)$ for the bidder $i$ is defined as
\begin{equation*}
    \max_{b_i'\in\mathcal{V}^m} u_i(v_i,(b_i',v_{-i}),x,y)-u_i(v_i,v,x,y).
\end{equation*}
An auction mechanism $(g,p)$ is DSIC, if and only if $\sum_{i=1}^n {reg}_i(v,x,y)=0$ for any value profile $v\in\mathcal{V}^{n\times m}$, bidder-context $x\in\mathcal{X}^n$, and item-context $y\in\mathcal{Y}^m$. 
Suppose the payment rule $p$ satisfies $p_i(b,x,y)\le \sum_{j=1}^mg_{i j}(b,x,y)b_{i j}$, which implies that each bidder has a non-negative utility. Then, the auction design can be modeled as a linear programming problem that maximizes the expected revenue  $\mathbb{E}_{(v,x,y)}[\sum_{i=1}^n p_i(v,x,y)]$ subject to the expected ex-post regret $\mathbb{E}_{(v,x,y)}[\sum_{i=1}^n reg_i(v,x,y)]=0$. Without loss of generality, the ex-post regret may refer to the average of all bidders' ex-post regrets. 
\begin{definition}
Suppose the network's parameter is $\omega$, and the bidder $i$'s empirical payment and ex-post regret are defined as 
$\frac{1}{L}\sum_{l=1}^L p^\omega_i(v^{(l)},x^{(l)},y^{(l)}),$
and
$\widehat{reg}_i(\omega)=\frac{1}{L}\sum_{l=1}^Lreg_i^\omega(v^{(l)},x^{(l)},y^{(l)})$, 
where the sample set $\{(v^{(l)},x^{(l)},y^{(l)})\}_{l=1}^L$ is {\it i.i.d.} sampled from the following prior distribution, 
$\mathbb{P}(v,x,y)=\prod_{i,j=1}^{n,m}\mathbb{P}(v_{ij}|x_i,y_j)\mathbb{P}_{X_i}(x_i)\mathbb{P}_{Y_j}(y_j)$.
\end{definition}


\textbf{Equivariant mapping.}
We define a mapping $f$ as $G$-{\it equivariant} if 
$\psi_g\circ f=f\circ\rho_g$
for two chosen group linear representations $\rho$ and $\psi$ and any $g$ in group $G$.
\begin{definition}[Permutation-Equivariant Mapping]
A permutation-equivariant mapping is defined to be $f:\mathbb{R}^{n\times m}\to \mathbb{R}^{n\times m}$ that for any instance $x\in\mathbb{R}^{n\times m}$, and permutation matrices $\sigma_n\in\mathbb{R}^{n\times n}$ and $\sigma_m\in\mathbb{R}^{m\times m}$, we have $f(\sigma_nx\sigma_m)=\sigma_nf(x)\sigma_m$.
\end{definition}
In this paper, we consider the bidder-permutation $\sigma_n\in \mathbb{R}^{n\times n}$ and item-permutation $\sigma_m\in \mathbb{R}^{m\times m}$. Specifically, we define a mapping $f$ is bidder-symmetric or item-symmetric, if $f(\sigma_n x)=\sigma_n f(x)$ or $f(x \sigma_m)=f(x)\sigma_m$, respectively. Moreover, we define an auction mechanism $(g,p)$ as bidder-symmetric or item-symmetric, if the allocation rule $g$ and the payment rule $p$ are both bidder-symmetric or item-symmetric. 

\textbf{Orbit averaging.}
For any feature mapping $f:\mathcal{F} \to \mathcal{G}$, the orbit averaging $\mathcal{Q}$ on $f$ is defined as
   $ \mathcal{Q}f=\frac{1}{|G|}\sum_{g\in G}\psi_{g}^{-1}\circ f\circ \rho_g,$
where $\rho$ and $\psi$ are two chosen group representations acting on the feature spaces $\mathcal{F}$ and $\mathcal{G}$, respectively. 
Orbit averaging can project any mapping to be equivariant:
\begin{proposition}\label{pro}
Orbit averaging $\mathcal{Q}$ is a projection to the equivariant mapping space $\{f:\psi\circ f=f\circ \rho\}$, i.e., $\psi\circ\mathcal{Q}f=\mathcal{Q}f\circ\rho$ and $\mathcal{Q}^2=\mathcal{Q}$. In particular, if $f$ is already equivariant, then $\mathcal{Q}f=f$.
\end{proposition}
Moreover, $\mathcal{Q}u$ and $\mathcal{Q}{reg}$ refer to the utility and the ex-post regret induced by $\mathcal{Q}g$ and $\mathcal{Q}p$. For the simplicity, we denote the orbit averagings that modify the auction mechanism to be bidder-symmetric, item-symmetric, and bidder/item-symmetric by bidder averaging $\mathcal{Q}_1$, item averaging $\mathcal{Q}_2$, and {bidder-item aggregated averaging} $\mathcal{Q}_3$. {Besides, a detailed proof of the feasibility of the projected mechanisms can be found in Appendix \ref{A.1}.}

\textbf{Hypothesis complexity.} The generalizatbility to unseen data is usually measured by the generalization bound, which depends on the hypothesis set's complexity. To characterize the complexity of the hypothesis set, we introduce the following definitions of {\it covering number} $\mathcal{N}_{\infty,1}$ and its corresponding distance $l_{\infty,1}$. Based on the covering number, we can obtain a generalization bound in Theorem \ref{gen}.
\begin{definition}[$l_{\infty,1}$-distance]
Let $\mathcal{X}$ be a feature space and $\mathcal{F}$ a space of functions from $\mathcal{X}$ to $\mathbb{R}^n$. The $l_{\infty,1}$-distance on the space $\mathcal{F}$ is defined as 
$l_{\infty,1}(f,g)=\max_{x\in \mathcal{X}}(\sum_{i=1}^n|f_i(x)-g_i(x)|)$.
\end{definition}
\begin{definition}[Covering number]
Covering number $\mathcal{N}_{\infty,1}(\mathcal{F},r)$ is the minimum number of balls with radius $r$ that can cover $\mathcal{F}$ under $l_{\infty,1}$-distance.
\end{definition}

\section{Theoretical Results}\label{main}
This section presents the theoretical results. {For simplicity, we view $p=(p_1,\dots,p_n)^T$ as a $n\times 1$ matrix to present the prices the bidders should pay.} We first prove that the permutation-equivariance induces the same expected revenue and a smaller expected ex-post regret in Section \ref{revenue and regret}. 
Next in Section \ref{gener}, we prove that the permutation-equivariant mechanism space has a smaller covering number, which promises a smaller required sample complexity and a better generalization. Detailed proofs are omitted from the main text and given in supplementary materials due to space limitation. 

\subsection{Benefits for Revenue and Ex-Post Regret}\label{revenue and regret}

In this section, we discuss the benefits for the revenue and the ex-post regret in the conditions of bidder-symmetry and item-symmetry separately, and then discuss the benefits when both of them hold. Based on these results, we also study the learning process of non-permutation-equivariant neural networks for auction design.



\subsubsection{Benefits in the Bidder/Item-Symmetry Condition}\label{bidder}
When the bidders come from the same distribution, the joint valuation distribution $f$ is invariant under bidder-permutation, {\it i.e.} $f(\sigma_n v,\sigma_n x,y)=f(v,x,y)$ for any $\sigma_n\in S_n$. Meanwhile, when the items are indistinguishable, the joint distribution $f$ is invariant under item-permutation, {\it i.e.}, $f(v\sigma_m,x,y\sigma_m)=f(v,x,y)$ for any $\sigma_m\in S_m$. Both conditions do not always hold simultaneously. In this section, we study them separately. 

To measure the ``non-permutation-equivariance'' of the mechanism, we introduce the conception of {\it regret gap} between the projected mechanism and the original mechanism as below,
\begin{equation*}
	\Delta_\cdot(g,p;v,x,y)=\max_{v'\in\mathcal{V}^{n\times m}}\sum_{i=1}^nu_i(v_i,(v_i',v_{-i}),x,y)-\max_{v'\in\mathcal{V}^{n\times m}}\sum_{i=1}^n[\mathcal{Q}_\cdot{u}]_i(v_i,(v_i',v_{-i}),x,y),
\end{equation*}
{where $v$ is the valuation profiles, $v_i$ is the valuation profile of bidder $i$, $x$ is the bidder-context, $y$ is the item-context, and the orbit averaging $\mathcal{Q}_\cdot$ can be the bidder averaging $\mathcal{Q}_1$ or the item averaging $\mathcal{Q}_2$.}

The bidder averaging $\mathcal{Q}_1$ and the item averaging $\mathcal{Q}_2$ acting on the allocation rule $g$ and the payment rule $p$, respectively, are as below,
\begin{gather*}
	\mathcal{Q}_1{g}(v,x,y)=\frac{1}{n!}\sum_{\sigma_n\in S_n}\sigma_n^{-1}g(\sigma_n v,\sigma_n x,y),~~ \mathcal{Q}_1{p}(v,x,y)=\frac{1}{n!}\sum_{\sigma_n\in S_n}\sigma_n^{-1}p(\sigma_n v,\sigma_n x,y),\\
\mathcal{Q}_2{g}(v,x,y)=\frac{1}{m!}\sum_{\sigma_m\in S_m}g(v \sigma_m,x,y\sigma_m)\sigma_m^{-1}, ~~\text{and}~~ \mathcal{Q}_2{p}(v,x,y)=\frac{1}{m!}\sum_{\sigma_m\in S_m}p(v\sigma_m,x,y\sigma_m).
\end{gather*}

We thus can prove the following theorem that characterizes the benefits of permutation-equivariance for revenue and ex-post regret in the condition of bidder/item-symmetry.
\begin{theorem}[Benefits for revenue and ex-post regret in the condition of bidder/item-symmetry]\label{the}
	When the valuation distribution is invariant under permutations of bidders/items, the projected mechanism has the same expected revenue and a smaller expected ex-post regret, that is,
	\begin{gather}\label{equ1}
		\mathbb{E}_{(v,x,y)}\bigg[\sum\limits_{i=1}^n [\mathcal{Q}_\cdot{p}]_i(v,x,y)\bigg]=
		\mathbb{E}_{(v,x,y)}\bigg[\sum\limits_{i=1}^n{p}_i(v,x,y)\bigg], ~~\text{and}\\
		\label{equ2}
		\mathbb{E}_{(v,x,y)}\bigg[\sum\limits_{i=1}^n{reg}_i(v,x,y)\bigg]-
		\mathbb{E}_{(v,x,y)}\bigg[\sum\limits_{i=1}^n[\mathcal{Q}_\cdot{reg}]_i(v,x,y)\bigg]
		=\mathbb{E}_{(v,x,y)}\Big[\Delta_\cdot(g,p;v,x,y)\Big]\ge 0,
	\end{gather}
	where $p$ is the payment rule, $reg$ is the ex-post regret, and $\mathcal{Q}_\cdot$ is the bidder/item averaging.
\end{theorem}


A smaller expected ex-post regret implies this mechanism is closer to the {\it dominant strategy incentive compatible} condition. Conversely, and equivalently, when the expected ex-post regrets are fixed, the projected auction mechanism has a larger expected revenue. 
For any auction mechanism, in the bidder/item-symmetry condition, we can project it through the bidder/item averaging. 

\begin{remark}

The mechanism space can be decomposed into the direct sum of the permutation-equivariant mechanism space $\{\mathcal{M}:\mathcal{Q}\mathcal{M}=\mathcal{M}\}$ and the complementary space $\{\mathcal{N}:\mathcal{Q}\mathcal{N}=0\}$ \cite{elesedy2021provably}. Thus, a mechanism $\mathcal{M}$ has a unique decomposition: $\mathcal{M}=\mathcal{Q}\mathcal{M}+\mathcal{N}$. The pure permutation-equivariant part $\mathcal{Q}\mathcal{M}$ contains all and only the ``permutation-equivariance'' 
of the mechanism $\mathcal{M}$. The pure non-permutation-equivariant part $\mathcal{N}$ is independent from the permutation-equivairance. 
In this way, we may study the influence of permutation-equivariance by comparing the mechanism $\mathcal{M}$ and its permutation equivariant part $\mathcal{Q}\mathcal{M}$.
\end{remark}

\subsubsection{Interplay between Bidder-Symmetry and Item-Symmetry.}\label{modify}
If the valuation distribution is invariant under both bidder-permutation and item-permutation, we can project the mechanism to be permutation-equivariant with respect to both bidder and item in two steps
(by mapping $\mathcal{Q}_1\circ\mathcal{Q}_2$ or mapping $\mathcal{Q}_2\circ\mathcal{Q}_1$). Consequently, the projected mechanism has the same expected revenue and a smaller expected ex-post regret. Equivalently, we can also project an auction mechanism to be bidder-symmetric and item-symmetric immediately by the bidder-item aggregated averaging $\mathcal{Q}_3$ as below,
\begin{gather*}
    \mathcal{Q}_3{g}(v,x,y)=\frac{1}{n!m!}{\sum\limits_{\sigma_n\in S_n}\sum\limits_{\sigma_m\in S_m}}\sigma_n^{-1}g(\sigma_n v \sigma_m, \sigma_n x,y\sigma_m)\sigma_m^{-1},~~\text{and}\\
    \mathcal{Q}_3{p}(v,x,y)=\frac{1}{n!m!}{\sum\limits_{\sigma_n\in S_n}\sum\limits_{\sigma_m\in S_m}}\sigma_n^{-1}p(\sigma_n v \sigma_m, \sigma_n x,y\sigma_m).
\end{gather*}

We can prove that the bidder-item aggregated averaging $\mathcal{Q}_3$ is the composition of the orbit averaging operators $\mathcal{Q}_1$ and $\mathcal{Q}_2$, as shown in the following lemma. This lemma shows that the order of $\mathcal{Q}_1$ and $\mathcal{Q}_2$ would not influence their composition.
\begin{lemma}\label{lemma:composition}
The bidder-item aggregated averaging is the composition of bidder averaging and item averaging:
    $\mathcal{Q}_3=\mathcal{Q}_1\circ \mathcal{Q}_2=\mathcal{Q}_2\circ \mathcal{Q}_1.$
\end{lemma}

Based on this lemma, we can prove the following theorem on the benefits of permutation-equivariance for revenue and ex-post regret in the condition of both bidder-symmetry and item-symmetry.

\begin{theorem}[Benefits for revenue and ex-post regret in the condition of both bidder-symmetry and item-symmetry]\label{theorem plus}
When the valuation distribution is invariant under both item-permutation and item-permutation, then the projected mechanism has a same expected revenue and a smaller expected ex-post regret, that is,
\begin{gather*}
    \mathbb{E}_{(v,x,y)}\bigg[\sum\limits_{i=1}^n \mathcal{Q}_3{p}_i(v,x,y)\bigg]=
    \mathbb{E}_{(v,x,y)}\bigg[\sum\limits_{i=1}^n{p}_i(v,x,y)\bigg] ~~\text{and}\\
    \mathbb{E}_{(v,x,y)}\bigg[\sum\limits_{i=1}^n{reg}_i(v,x,y)\bigg]-\mathbb{E}_{(v,x,y)}\bigg[\sum\limits_{i=1}^n[\mathcal{Q}_3{reg}]_i(v,x,y)\bigg]=\mathbb{E}_{(v,x,y)}\Big[\Delta_3(g,p;v,x,y)\Big]\ge 0,
\end{gather*}
where $p$ is the payment rule, $reg$ is the ex-post regret, and $\mathcal{Q}_3$ is the bidder-item aggregated averaging.
\end{theorem}

The difference between bidder-symmetry and item-symmetry is significant in practice. For example, for a symmetric valuation distribution, when the mechanism is already bidder-symmetric but not item-symmetric, we can project it to be item-symmetric to gain an extra benefit from item-symmetry. That means, the two regret gaps induced by $
\mathcal{Q}_1$ and $\mathcal{Q}_2$ are ``additive'' as below,
\begin{align*}
    &\Delta_3(g,p;v,x,y)
    =\Delta_1(g,p;v,x,y)+\Delta_2(\mathcal{Q}_1g,\mathcal{Q}_1p;v,x,y).
\end{align*}

In general, 
    $\mathbb{E}[\Delta_2(g,p;v,x,y)]\not=\mathbb{E}[\Delta_2(\mathcal{Q}_1g,\mathcal{Q}_1p;v,x,y)]$
and thus 
    $\mathbb{E}[\Delta_3(g,p;v,x,y)]\not=\mathbb{E}[\Delta_1(g,p;v,x,y)]+\mathbb{E}[\Delta_2(g,p;v,x,y)].$
Thus, the benefits from bidder-symmetry and item-symmetry are ``additive'' but not strictly ``independent''.

\subsubsection{Insights on Training Non-Permutation-Equivariant Mechanism}
Because the expected revenue is always the same for the original mechanism and  the projected permutation-equivariant mechanism, we only consider the gradient caused by the expected ex-post regret. We can decompose the original expected ex-post regret into the sum of the expected ex-post regret of the projected mechanism and the expectation of the regret gap as below,
\begin{equation*}
	\mathbb{E}_{(v,x,y)}\bigg[\sum\limits_{i=1}^n{reg}_i(v,x,y)\bigg]=\mathbb{E}_{(v,x,y)}\bigg[\sum\limits_{i=1}^n[\mathcal{Q}_{3}{reg}]_i(v,x,y)\bigg]+\mathbb{E}_{(v,x,y)}\Big[\Delta_{3}(g,p;v,x,y)\Big].
\end{equation*}

The regret gap $\Delta_{3}(g,p;\cdot)$ follows from the ``non-permutation-equivariance'' of the mechanism $\mathcal{M}$. When the distance $l(\mathcal{M},\mathcal{Q}\mathcal{M})$ tends to $0$, the regret gap converges to $0$. 
When the auction mechanism has a negligible ex-post regret, the expectation of the regret gap is also close to $0$. That means, the mechanism is close to being permutation-equivariant.
However, even using a symmetric dataset or adopting data augmentation in training, the learned mechanism will not be permutation-equivariant in general \cite{lyle2020benefits}. As a result, to achieve negligible ex-post regret, the non-permutation-equivariant models need to learn more samples to approach permutation-equivariance. That is because the non-permutation-equivariant part (expected regret gap) would mislead the gradient of the expected regret but have no benefit to the expected revenue and the expected ex-post regret.

On the other hand, the regret gap can be viewed as a regularizer in the ex-post regret to penalize the ``non-permutation-equivariance'' of the mechanism. When the optimizer tries to minimize the ex-post regret, the auction mechanism approaches to be permutation-equivariant. Therefore, if the mechanism achieves a negligible ex-post regret, it is almost to be permutation-equivariant. This result can explain why RegretNet struggles to find permutation-equivariant auction mechanisms \citep{rahme2021permutation}. However, in complex settings, it will be harder for non-permutation-equivariant models to approach the negligible ex-post regret.  It can explain why the permutation-equivariant models show a significant improvement in complex settings, compared with that they have similar performances in simple settings \cite{duan2022context,ivanov2022optimal}, which shows the great importance of adopting permutation-equivariant models in complex settings.

\subsection{Benefits for Generalization}\label{gener}
In this section, we study permutation-equivariance from the aspect of generalizability \cite{mohri2018foundations, he2020recent}, which characterizes the performance gap of a learned mechanism on collected training data and unseen data. 


We first study the covering number of the permutation-equivariant mechanism space. {Let $\mathcal{U}=\{u^\omega:\omega\in\Omega\}$ and $\mathcal{P}=\{p^\omega:\omega\in\Omega\}$ be the spaces of all possible utilities and payment rules, and $\mathcal{Q}_\cdot{\mathcal{U}}=\{\mathcal{Q}_\cdot{u}:u\in\mathcal{U}\}$ and $\mathcal{Q}_\cdot{\mathcal{P}}=\{\mathcal{Q}_\cdot{p}:p\in\mathcal{P}\}$ the spaces of all projected utilities and payment rules. In addition, let $\mathcal{N}_{\infty,1}({\mathcal{U}},r)$ and $\mathcal{N}_{\infty,1}({\mathcal{P}},r)$ be the minimum numbers of balls with radius $r$ that can cover $\mathcal{U}$ and $\mathcal{P}$ under $l_{\infty,1}$-distance, respectively.} We obtain the following result, which indicates the projected permutation-equivariant mechanism space has smaller covering numbers. 
\begin{theorem}[Covering number of the permutation-equivariant mechanism space]\label{covering}
The space of all projected bidder-symmetric mechanisms has smaller covering numbers, that is,
\begin{equation*}
    \mathcal{N}_{\infty,1}(\mathcal{Q}_1{\mathcal{U}},r)\le\mathcal{N}_{\infty,1}(\mathcal{U},r)
~~\text{   and }~~
    \mathcal{N}_{\infty,1}(\mathcal{Q}_1{\mathcal{P}},r)\le\mathcal{N}_{\infty,1}(\mathcal{P},r).
\end{equation*}
The space of all projected item-symmetric mechanisms has smaller covering numbers, that is,
\begin{equation*}\label{covering2}
    \mathcal{N}_{\infty,1}(\mathcal{Q}_2{\mathcal{U}},r)\le\mathcal{N}_{\infty,1}(\mathcal{U},r)
~~\text{and}~~
    \mathcal{N}_{\infty,1}(\mathcal{Q}_2{\mathcal{P}},r)\le\mathcal{N}_{\infty,1}(\mathcal{P},r).
\end{equation*}
\end{theorem}
Intuitively, the orbit averaging $\mathcal{Q}$ narrows the distance between two mechanisms: $l(\mathcal{Q}\mathcal{M},\mathcal{Q}\mathcal{M}')\le l(\mathcal{M},\mathcal{M}')$, for any two mechanisms. Then, any $r$-cover $\mathcal{A}$ for space $\mathcal{U}$ or space $\mathcal{P}$ induces an $r$-cover $\mathcal{Q}\mathcal{A}$ for space $\mathcal{Q}\mathcal{U}$ or space $\mathcal{Q}\mathcal{P}$.

Combining with Lemma \ref{lemma:composition}, we have the following results,
	\begin{gather*}
		\mathcal{N}_{\infty,1}(\mathcal{Q}_3\mathcal{U},r)=\mathcal{N}_{\infty,1}(\mathcal{Q}_1\mathcal{Q}_2\mathcal{U},r)\le\mathcal{N}_{\infty,1}(\mathcal{Q}_2\mathcal{U},r)\le\mathcal{N}_{\infty,1}(\mathcal{U},r)~~\text{and}\\
		\mathcal{N}_{\infty,1}(\mathcal{Q}_3\mathcal{P},r)=\mathcal{N}_{\infty,1}(\mathcal{Q}_1\mathcal{Q}_2\mathcal{P},r)\le\mathcal{N}_{\infty,1}(\mathcal{Q}_2\mathcal{P},r)\le\mathcal{N}_{\infty,1}(\mathcal{P},r).
	\end{gather*}

We then prove that two generalization bounds of permutation-equivariant mechanisms, which characterize the gap between the expected revenue/ex-post regret and their empirical counterparts. Similar generalization results are existing in previous works \cite{duan2022context,dutting2019optimal}.

\begin{theorem}[Generalization bounds of permutation-equivariant mechanisms]\label{gen}
If for any bidder, her valuation satisfies that $v_i(S)\le 1$ for any $S\subset[m]$, then with probability at least $1-\delta$, we have the following inequalities with $\epsilon\ge\sqrt{\frac{9n^2}{2L}(\log\frac{4}{\delta}+\max\{\log\mathcal{N}_{\infty,1}(\mathcal{P},\frac{\epsilon}{3}),log\mathcal{N}_{\infty,1}(\mathcal{U},\frac{\epsilon}{6})\})}$,
\begin{gather}
\label{eq:gen1}
    \bigg|\mathbb{E}\bigg[\sum\limits_{i=1}^n p_i^{\omega}(v,x,y)\bigg]-\frac{1}{L}\sum\limits_{l=1}^L\sum\limits_{i=1}^n p_i^{\omega}(v^{(l)},x^{(l)},y^{(l)})\bigg|\le\epsilon,~~\text{and}\\
\label{eq:gen2}
   \bigg|\mathbb{E}\bigg[\sum\limits_{i=1}^n {reg}_i^{\omega}(v,x,y)\bigg]-\sum\limits_{i=1}^n \widehat{reg}_i(\omega)\bigg|\le\epsilon,
\end{gather}
where $L$ is the number of samples, $\mathcal{U}$ and $\mathcal{P}$ are the spaces of all possible utilities and payment rules.  
\end{theorem}

Equivalently, we can rewrite this result in the form of the sample complexity,
\begin{corollary}\label{corollary}
For any $\epsilon>0$, $\delta\in (0,1)$, and mechanism parameter $\omega$, when the sample complexity $L\ge \frac{9n^2}{2\epsilon^2}\Big(\log\frac{4}{\delta}+\max\Big\{\log\mathcal{N}_{\infty,1}(\mathcal{P},\frac{\epsilon}{3}),log\mathcal{N}_{\infty,1}(\mathcal{U},\frac{\epsilon}{6})\Big\}\Big)$, with probability at least $1-\delta$, the generalization bounds, eqs. (\ref{eq:gen1}) and (\ref{eq:gen2}), hold.
\end{corollary}

\begin{remark}
Combining Theorem \ref{covering}, we have proved that the permutation-equivariance can improve the generalizability.
\end{remark}


\section{Experiments}\label{ex}
This section presents our experimental results. More details and results are presented in the supplementary materials.

\textbf{Model architecture.}\label{impl}
We project RegretNet \citep{dutting2019optimal} to the permutation-equivariant mechanism space via 
employing bidder-item aggregated averaging for the bidder-symmetry and item-symmetry condition. The projected model is called RegretNet-PE. {We also project the well-trained RegretNet, called RegretNet-Test.} Specifically, RegretNet is an auction mechanism defined as $(g^{\omega}, p^{\omega})$, in which both the allocation rule $g^{\omega}$ and the payment rule $p^{\omega}$ are neural networks that consist of three fully-connected layers, and $\omega$ is the overall model parameter of the auction mechanism. The detailed architecture is given in the supplementary materials. 

{\textbf{Comparison with EquivariantNet.} 
RegretNet uses two feed-forward fully-connected networks to learn the allocation rule and payment rule, respectively. We denote the weight matrix in the layer $\ell$ as $W^{(\ell)}$. Both EquivariantNet and RegretNet-PE inherit the architecture of RegretNet (with some modifications), but utilize different approaches to realize the permutation-equivariance. EquivariantNet applies parameter-sharing in every layer during training, to constrain $W^{(\ell)}$ to be equivariant. In contrast, RegretNet-PE employs orbit averaging to be permutation-equivariant. Specifically, RegretNet-PE adopts a weight matrix $I_{K}\otimes W^{(\ell)}(\rho^T_{g_{1}}\dots\rho^T_{g_K})^T$ in the first layer, weight matrices $I_{K}\otimes W^{(\ell)}$ in the following layers, and multiples a matrix $(\rho_{g_1}^{-1},\dots,\rho_{g_K}^{-1})$ to the output layer, where $K$ is the scale of the group $G=\{g_1,\dots,g_K\}$, $\rho_{g_k}$ represents the permutation operator on bidders and items, $I_K$ is an identity matrix, and $\otimes$ is the Kronecker product. It is worth noting that RegretNet-PE is only designed for verifying our theory.}

\begin{table}[t]
\centering{
\caption{{Experimental results. "$n\times m$ Uniform" refers that there are $n$ bidders and $m$ items, and the valuations are i.i.d. drawn from the uniform distribution $U[0,1]$.} To simplify, we multiply all results by a factor of $10^5$ for the ex-post regret and generalization error (GE).}
\label{table:single_item}
\setlength{\tabcolsep}{1.6mm}
\renewcommand{\arraystretch}{1.2}
\begin{tabular}{c|ccc|ccc|ccc}
\toprule
\multirow{2}{*}{Method} & \multicolumn{3}{c|}{$2\times1$ Uniform}            & \multicolumn{3}{c|}{$3\times1$ Uniform}   & \multicolumn{3}{c}{$5\times1$ Uniform}        \\
& Revenue & Regret & GE & Revenue & Regret & GE   & Revenue & Regret & GE\\          
\midrule
Optimal                 & ${0.417}$   & 0    & -  & ${0.531}$  & 0   & -  &  ${0.672}$       & 0      & -  \\
\midrule
RegretNet               & {$0.415$}   & {$17.4$} & {$6.00$} & ${0.535}$ & ${18.3}$ & $11.4$ &  $0.658$  & ${15.9}$ & ${6.40}$\\
RegretNet-Test          & ${0.415}$   & {$16.3$} & -    & $0.535$ &$13.3$ & - & $0.658$   & $6.50$ & - \\
RegretNet-PE            & {$0.420$}   & {$14.6$} & ${3.90}$  & {$0.541$} &
${16.4}$ & {$10.2$} &  ${0.677}$  & ${13.2}$ & ${5.10}$\\
\bottomrule

\end{tabular}
}
\end{table}
\begin{table}[t]
    \centering{
    \caption{{Experimental results. "$n\times m$ Uniform" refers that there are $n$ bidders and $m$ items, and the valuations are i.i.d. sampled from the uniform distribution $U[0,1]$.}}
    \label{table+}
	\setlength{\tabcolsep}{1.6mm}
	\renewcommand{\arraystretch}{1.2}
		\begin{tabular}{c|cc|cc}
			\toprule
			\multirow{2}{*}{Method} & \multicolumn{2}{c|}{$1\times2$ Uniform}     &   \multicolumn{2}{c}{$2\times 2$ Uniform}    \\
			& Revenue & Regret & Revenue & Regret\\          
			\midrule
		    RegretNet           & $0.562$   & $0.00061$   & $0.870$ & $0.00070$            \\
			EquivariantNet      & $0.551$   & $0.00013$     & $0.873$ & $0.00100$             \\
			RegretNet-Test      & $0.562$   & $0.00052$     & $0.870$ & $0.00054$             \\
			RegretNet-PE        & $0.563$   & $0.00037$     & $0.913$ & $0.00067$   \\
			\bottomrule
		\end{tabular}
    }
\end{table}

\textbf{Auction settings.}
{We first adopt the two-bidder single-item, two-bidder single-item, three-bidder single-item, and five-bidder single-item settings in the experiments that compare the learned mechanisms with theoretical optimal mechanisms.}
The optimal auction mechanism for any single-item auction is known \citep{myerson1981optimal}. We thus compare the mechanisms leaned by our method with the optimal auction mechanisms in the single-item settings. {Also, we compare RegretNet-PE and EquivariantNet in the one-bidder, two-item setting, and the two-bidder, two-item setting.}
Besides, we employ a multivariate uniform distribution $U[0,1]^m$ to model the bidder valuation profiles. 
{ In all settings, we sample 640,000 data points for training and 5,000 points for test.} Due to the space limitation, we place the results of two-bidder five-item and five-bidder three-item settings in Appendix \ref{B.2}.



\textbf{Model training.}
We optimize the auction mechanism model via solving the following optimization problem, following the standard settings \cite{dutting2019optimal, rahme2021permutation, duan2022context,ivanov2022optimal},
\begin{equation*}\label{lag}
    \mathcal{L}_{\rho}(\omega,\lambda)=-\frac{1}{L}\sum\limits_{l=1}^L\sum\limits_{i=1}^n p_i^\omega(v^{(l)},x^{(l)},y^{(l)})+\sum\limits_{i=1}^n\lambda_i\widehat{reg}_i(\omega)+\frac{\rho}{2}\Big(\sum\limits_{i=1}^n\widehat{reg}_i(\omega)\Big)^2,
\end{equation*}
where $\lambda\in\mathbb{R}^n$ is the Lagrange multiplier and $\rho>0$ is the factor of the quadratic regularization term.
During the training process, the objective function $\mathcal{L}_{\rho}(\omega,\lambda)$ is minimized via Adam with a learning rate of $0.001$ with respect to the model parameter $\omega$ and the Lagrange multiplier $\lambda$ is updated once in every 100 iterations, until the ex-post regret is smaller than $0.001$. The regularization factor $\rho$ is set to $1.0$ initially and gradually increased along the training process.
In calculating the best bid profile $v_i'$ of every bidder $i$, we first randomly initialize the bid profiles once in training and 1,000 times in test, optimize each of them individually via Adam with the same settings, and take the best one as the approximated best bidding.

\textbf{Evaluation.}
We leverage three metrics to evaluate the performance of the auction mechanism, which are:
(1) the empirical revenue $\widehat{rev}$,
(2) the empirical ex-post regret averaging across all bidders, {\it i.e.}, $\widehat{reg}=\frac{1}{n}\sum_{i=1}^n\widehat{reg}_i$,
and (3) the generalization error defined on top of regrets, {\it i.e.}, $GE = \lvert \widehat{reg}_{test} - \widehat{reg}_{train}\rvert$, where $\widehat{reg}_{test}$ and $\widehat{reg}_{train}$ are the empirical ex-post regrets during test and training, respectively.

\textbf{Computing resource. } The experiments are conducted on 1 GPU (NVIDIA\textsuperscript{\textregistered} Tesla\textsuperscript{\textregistered} V100 16GB) and 10 CPU cores (Intel\textsuperscript{\textregistered} Xeon\textsuperscript{\textregistered} Processor E5-2650 v4 @ 2.20GHz).

\textbf{Experimental results.}
We train a RegretNet and a RegretNet-PE on the training data. The well-trained RegretNet is then projected to be permutation-equivariant, denoted as ``RegretNet-Test''.
The results are collected in Tables {\ref{table:single_item} and \ref{table+}.} 

{
From Tables \ref{table:single_item} and \ref{table+}, we observe that 
(1) compared to RegretNet, RegretNet-PE has a significantly higher revenue with a lower ex-post regret, and narrows the generalization gap between the training ex-post regret and its test counterpart; (2) compared to RegretNet, RegretNet-Test receives the same revenue with a significantly lower ex-post regret; and (3) under comparable ex-post regrets, RegretNet-PE has considerably higher revenue than EquivariantNet, while all permutation-equivariant models (RegretNet-Test and RegretNet-PE) can outperform RegretNet. These results show significant benefits of permutation-equivariance on revenue, ex-post regret, and generalizability, which fully supports our theoretical findings in Theorems \ref{the}, \ref{theorem plus}, and \ref{covering}.}

\section{Conclusion and future works}\label{con}
In this paper, under {\it additive valuation} and {\it symmetric valuation} setting, we study the benefits of permutation-equivariance in auction mechanisms in two aspects: a better performance and a better generalization. First, we prove a smaller expected ex-post regret and the same expected revenue when projecting a mechanism to be permutation-equivariant. Next, we propose the permutation-equivariant mechanism space has a smaller covering number, which promises the permutation-equivariant models a better generalization. Extensive experiments are conducted to verify our theoretical results. Our results help understand the optimal auction mechanisms' characterization and the learning processes difference between non-equivariant models and equivariant models.

Beyond the additive valuation setting, an interesting direction is to extend our results to other conditions, including the combinatorial and the unit-demand auctions. Meanwhile, the understanding of the difference in the aspect of the training process between non-equivariant models and equivariant models is still elusive. 

\textbf{Social impact.}
Our results can help understand and design optimal auction mechanisms for symmetric valuation distribution. As a result, our work could inspire more near-optimal auction mechanisms and promote economic growth. No potential negative social impact is identified.

\begin{ack}
This work was supported in part by the Major Science and Technology Innovation 2030 ``New Generation Artificial Intelligence'' Key Project (No. 2021ZD0111700), the National Natural Science Foundation of China (No. 62006137), and the Beijing Outstanding Young Scientist Program (No. BJJWZYJH012019100020098).

We sincerely appreciate Hang Yu, Xiaowen Wei, Kaifan Yang, Shaopeng Fu, and Qingsong Zhang for the valuable comments and the anonymous NeurIPS reviewers for the helpful feedback.
\end{ack}

\section*{Checklist}
\begin{enumerate}
\item For all authors...
\begin{enumerate}
  \item Do the main claims made in the abstract and introduction accurately reflect the paper's contributions and scope?
    \answerYes{} Our contributions are stated in the 4th-8th sentences of the abstract, and the 4th-6th paragraphs of the introduction. The scope is described in the related work section. 
  \item Did you describe the limitations of your work?
    \answerYes{} We describes the problem setting we considered in the 3rd paragraph of the introduction. Our contributions are for this problem setting.
  \item Did you discuss any potential negative societal impacts of your work?
    \answerYes{} Our work is committed to learn the optimal auction mechanism. 
    No negative societal impact is identified.
  \item Have you read the ethics review guidelines and ensured that your paper conforms to them?
    \answerYes{} We have read the ethics review guidelines. We ensure that our paper conforms all of them.
\end{enumerate}
\item If you are including theoretical results...
\begin{enumerate}
  \item Did you state the full set of assumptions of all theoretical results?
    \answerYes{} We state all of our assumptions in the third paragraph of introduction, the first paragraph of preliminary, Theorem \ref{the}, and Theorem \ref{gener}. Wherever we use an assumption, we always state it again.
        \item Did you include complete proofs of all theoretical results?
    \answerYes{} Full proofs are given in the supplemental material.
\end{enumerate}
\item If you ran experiments...
\begin{enumerate}
  \item Did you include the code, data, and instructions needed to reproduce the main experimental results (either in the supplemental material or as a URL)?
    \answerNo{} 
    We are associated in a industrial organization. The code is currently under review of our organization.
  \item Did you specify all the training details (e.g., data splits, hyperparameters, how they were chosen)?
    \answerYes{} We mostly follow the setting of the RegretNet. Other details are given in Sections \ref{impl} and \ref{exdetails}.
        \item Did you report error bars (e.g., with respect to the random seed after running experiments multiple times)?
    \answerNo{}
        \item Did you include the total amount of compute and the type of resources used (e.g., type of GPUs, internal cluster, or cloud provider)?
    \answerYes{} They are given in Section \ref{impl}.
\end{enumerate}
\item If you are using existing assets (e.g., code, data, models) or curating/releasing new assets...
\begin{enumerate}
  \item If your work uses existing assets, did you cite the creators?
    \answerYes{} They are given in Section \ref{ex}.
  \item Did you mention the license of the assets?
\answerYes{} They are given in Section \ref{impl}.
\item Did you include any new assets either in the supplemental material or as a URL?
    \answerNo{}
    We are associated in a industrial organization. The code is currently under review of our organization.
  \item Did you discuss whether and how consent was obtained from people whose data you're using/curating?
\answerYes{} The data is synthetic in our data.
  \item Did you discuss whether the data you are using/curating contains personally identifiable information or offensive content?
\answerYes{} The data is synthetic in our data. No personally identifiable information or offensive content is contained.
\end{enumerate}
\item If you used crowdsourcing or conducted research with human subjects...
\begin{enumerate}
  \item Did you include the full text of instructions given to participants and screenshots, if applicable?
    \answerNA{}
  \item Did you describe any potential participant risks, with links to Institutional Review Board (IRB) approvals, if applicable?
    \answerNA{}
  \item Did you include the estimated hourly wage paid to participants and the total amount spent on participant compensation?
    \answerNA{}
\end{enumerate}
\end{enumerate}



\appendix
\section{Proofs}

This appendix collects all proofs omitted from the main text due to space limitation.
{\subsection{Proofs about Orbit Averaging}\label{A.1}
In this section, we prove Proposition \ref{pro} and the feasibility of the projected mechanisms. For simplicity, we use the notations $\sigma_n(i)$ and $\sigma_m(j)$ to present the ranks of the bidder $i$ and the item $j$ after bidder-permutation $\sigma_n$ and item-permutation $\sigma_m$, respectively.
\paragraph{Proof of Proposition \ref{pro}.} By the definition of orbit averaging $\mathcal{Q}$, we have
\begin{equation*}
    \mathcal{Q}f\circ\rho_g(x)=\frac{1}{|G|}\sum_{h\in G}\psi_h^{-1}f(\rho_h\rho_g x)=\psi_{g}\circ \frac{1}{|G|}\sum_{h\in G}\psi_{hg}^{-1}f(\rho_{hg} x)=\psi_g\circ \mathcal{Q}f(x).
\end{equation*}
In addition, if $f$ is equivariant, then we have
\begin{equation*}
    \mathcal{Q}f=\frac{1}{|G|}\sum_{g\in G}\psi_g^{-1}\circ f \circ \rho_g=\frac{1}{|G|}\sum_{g\in G}\psi_g^{-1}\circ \psi_g\circ f=\frac{1}{|G|}\sum_{g\in G} f=f.
\end{equation*}
Thus, orbit averaging is a projection to equivariant function space and fixes all equivariant functions. In addition, orbit averaging fixes all equivariant functions. That means, every equivariant function can be obtained by orbit averaging. In this sense, every equivariant models are contained in the orbit averaging framework.

\paragraph{Proof of the feasibility of projected mechanisms.}
We verify all feasibility conditions for the projected mechanisms as follows.

Firstly, for the allocation rule, we have
\begin{align*}
    &\sum_{i=1}^n(\mathcal{Q}_1g)_{i j}(v,x,y)=\frac{1}{n!}\sum_{i=1}^n\sum_{\sigma_n\in S_n}g_{\sigma_n^{-1}(i)j}(\sigma_n v,\sigma_n x,y)\\
    =&\frac{1}{n!}\sum_{\sigma_n\in S_n}\bigg[\sum_{i=1}^n g_{\sigma_n^{-1}(i)j}(\sigma_n v,\sigma_n x,y)\bigg]\le \frac{1}{n!}\sum_{\sigma_n\in S_n}1=1,
\end{align*}
\begin{align*}
    &\sum_{i=1}^n(\mathcal{Q}_2g)_{i j}(v,x,y)=\frac{1}{m!}\sum_{i=1}^n\sum_{\sigma_m\in S_m}g_{i\sigma_m^{-1}(j)}(v\sigma_m,x,y\sigma_m)\\
    =&\frac{1}{m!}\sum_{\sigma_m\in S_m}\bigg[\sum_{i=1}^n g_{i\sigma_m^{-1}(j)}(v\sigma_m,x,y\sigma_m)\bigg]\le \frac{1}{m!}\sum_{\sigma_m\in S_m}1=1,
\end{align*}
and
\begin{align*}
    &\sum_{i=1}^n(\mathcal{Q}_3g)_{i j}(v,x,y)=\frac{1}{n!m!}\sum_{i=1}^n\sum_{\sigma_n\in S_n}\sum_{\sigma_m\in S_m}g_{\sigma_n^{-1}(i)\sigma_m^{-1}(j)}(\sigma_n v\sigma_m,\sigma_n x,y\sigma_m)\\
    =&\frac{1}{n!m!}\sum_{\sigma_n\in S_n}\sum_{\sigma_m\in S_m}\bigg[\sum_{i=1}^n g_{\sigma_n^{-1}(i)\sigma_m^{-1}(j)}(\sigma_n v\sigma_m,\sigma_n x,y\sigma_m)\bigg]\le \frac{1}{n!m!}\sum_{\sigma_n\in S_n}\sum_{\sigma_m\in S_m}1=1,
\end{align*}
thus we know the projected allocation rule is also feasible, which will never allocate one item more than once.

In addition, for the payment rule, we have
\begin{align*}
    &(\mathcal{Q}_1 p)_{i}(v,x,y)=\frac{1}{n!}\sum_{\sigma_n\in S_n}p_{\sigma_n^{-1}(i)}(\sigma_n v,\sigma_n x,y)\\
    \le&\frac{1}{n!}\sum_{\sigma_n\in S_n}\bigg[\sum_{j=1}^m g_{\sigma_{n}^{-1}(i)j}(\sigma_n v,\sigma_n x,y)(\sigma_n v)_{\sigma_n^{-1}(i) j}\bigg]\\
    =&\sum_{j=1}^m\bigg[\frac{1}{n!}\sum_{\sigma_n\in S_n}g_{\sigma_{n}^{-1}(i)j}(\sigma_n v,\sigma_n x,y)v_{i j}\bigg]\le\sum_{j=1}^m(\mathcal{Q}_1 g)_{i j}(v,x,y)v_{i j},
\end{align*}
\begin{align*}
    &(\mathcal{Q}_2 p)_{i}(v,x,y)=\frac{1}{m!}\sum_{\sigma_m\in S_m}p_{i}(v\sigma_m,x,y\sigma_m)\\
    \le&\frac{1}{m!}\sum_{\sigma_m\in S_m}\bigg[\sum_{j=1}^m g_{i j}(v\sigma_m, x,y\sigma_m)(v\sigma_m)_{i j}\bigg]\\
    =&\sum_{j=1}^m\bigg[\frac{1}{m!}\sum_{\sigma_m\in S_m}g_{i \sigma_m^{-1}(j)}(v\sigma_m, x,y\sigma_m)v_{i j}\bigg]\le\sum_{j=1}^m(\mathcal{Q}_2 g)_{i j}(v,x,y)v_{i j},
\end{align*}
and
\begin{align*}
    &(\mathcal{Q}_3 p)_{i}(v,x,y)=\frac{1}{n!m!}\sum_{\sigma_n\in S_n}\sum_{\sigma_m\in S_m}p_{\sigma_n^{-1}(i)}(\sigma_n v\sigma_m,\sigma_n x,y\sigma_m)\\
    \le&\frac{1}{n!m!}\sum_{\sigma_n\in S_n}\sum_{\sigma_m\in S_m}\bigg[\sum_{j=1}^m g_{\sigma_n^{-1}(i) j}(\sigma_n v\sigma_m, \sigma_n x,y\sigma_m)(\sigma_n v\sigma_m)_{\sigma_n^{-1}(i) j}\bigg]\\
    =&\sum_{j=1}^m\bigg[\frac{1}{n!m!}\sum_{\sigma_n\in S_n}\sum_{\sigma_m\in S_m}g_{\sigma_n^{-1}(i) \sigma_m^{-1}(j)}(\sigma_n v\sigma_m, \sigma_n x,y\sigma_m)v_{i j}\bigg]\le\sum_{j=1}^m(\mathcal{Q}_3 g)_{i j}(v,x,y)v_{i j},
\end{align*}
thus, we have completed this proof.
}
\subsection{Proof of Theorem \ref{the}}

In this section, we proves Theorem \ref{the}. 

\begin{proof}[Proof of Theorem \ref{the}.]



We first study the part in the condition of bidder-symmetry; {\it i.e.}, 
the orbit averaging $\mathcal{Q}_\cdot$ is the bidder averaging $\mathcal{Q}_1$, 
acting on the allocation rule $g$ and the payment rule $p$ as below,
\begin{equation*}
		\mathcal{Q}_1{g}(v,x,y)=\frac{1}{n!}\sum_{\sigma_n\in S_n}\sigma_n^{-1}g(\sigma_n v,\sigma_n x,y),
\end{equation*}
and
\begin{equation*}
\mathcal{Q}_1{p}(v,x,y)=\frac{1}{n!}\sum_{\sigma_n\in S_n}\sigma_n^{-1}p(\sigma_n v,\sigma_n x,y).
\end{equation*}

\paragraph{Step 1:} We first prove that the auction  mechanism has the same expected revenue after projection, {\it i.e.}, \begin{equation*}
	\underset{(v,x,y)}{\mathbb{E}}\bigg[\sum\limits_{i=1}^n[\mathcal{Q}_1p]_i(v,x,y)\bigg]=\underset{(v,x,y)}{\mathbb{E}}\bigg[\sum\limits_{i=1}^n p_i(v,x,y)\bigg].
\end{equation*}

{Given that for all permutation $\pi$,
\begin{equation*}
    \sum_{i=1}^np_i=\sum_{i=1}^n p_{\pi (i)},
\end{equation*}
we have the following equation,}
\begin{align*}
	&\underset{(v,x,y)}{\mathbb{E}}\bigg[\sum\limits_{i=1}^n \mathcal{Q}_1{p}_i(v,x,y)\bigg]
	=\underset{(v,x,y)}{\mathbb{E}}\bigg[\frac{1}{n!}\sum\limits_{i=1}^n\sum\limits_{\sigma_n\in S_n}{p}_{\sigma_n^{-1}(i)}(\sigma_n v,\sigma_n x,y)\bigg]\\
	=&\underset{(v,x,y)}{\mathbb{E}}\bigg[\frac{1}{n!}\sum\limits_{i=1}^n\sum\limits_{\sigma_n\in S_n}{p}_{i}(\sigma_n v,\sigma_n x,y)\bigg]
	=\frac{1}{n!}\sum\limits_{i=1}^n\sum\limits_{\sigma_n\in S_n} \underset{(v,x,y)}{\mathbb{E}}\bigg[{p}_i(\sigma_n v,\sigma_n x,y)\bigg]\\
	=&\frac{1}{n!}\sum\limits_{i=1}^n\sum\limits_{\sigma_n\in S_n} \underset{(v,x,y)}{\mathbb{E}}\bigg[{p}_i(v,x,y)\bigg]
	=\sum\limits_{i=1}^n\underset{(v,x,y)}{\mathbb{E}}\bigg[{p}_i(v,x,y)\bigg].
\end{align*}

{Thus, we complete the first step.}

\paragraph{Step 2:} We then prove that 
the sum of all bidders' utilities remains the same after projection, {\it i.e.}, \begin{equation*}
	\underset{(v,x,y)}{\mathbb{E}}\bigg[\sum_{i=1}^n[\mathcal{Q}_1u]_i(v_i,v,x,y)\bigg]=\underset{(v,x,y)}{\mathbb{E}}\bigg[\sum_{i=1}^n u_i(v_i,v,x,y)\bigg].
\end{equation*}

Given that
\begin{equation*}
    \sum_{i=1}^n g_{\pi(i)j}v_{\pi(i)j}=\sum_{i=1}^n g_{i j}v_{i j},
\end{equation*} 
we have the following equation,
	\begin{align*}
		&\underset{(v,x,y)}{\mathbb{E}}\bigg[\sum\limits_{i=1}^n[\mathcal{Q}_1{u}]_i(v_i,v,x,y)\bigg]\\
		=&\underset{(v,x,y)}{\mathbb{E}}\bigg[\sum\limits_{i=1}^n\Big[\sum\limits_{j=1}^m [\mathcal{Q}_1{g}]_{i j}(v,x,y)v_{i j}-[\mathcal{Q}_1{p}]_i(v,x,y)\Big]\bigg]\\
		=&\underset{(v,x,y)}{\mathbb{E}}\bigg[\sum\limits_{i=1}^n\sum\limits_{j=1}^m[\mathcal{Q}_1{g}]_{i j}(v,x,y)v_{i j} \bigg]-\underset{(v,x,y)}{\mathbb{E}}\bigg[\sum\limits_{i=1}^n{p}_i(v,x,y)\bigg]\\
		=&\underset{(v,x,y)}{\mathbb{E}}\bigg[\frac{1}{n!}\sum\limits_{i=1}^n\sum\limits_{j=1}^m\sum\limits_{\sigma_n\in S_n}{g}_{\sigma_n^{-1}(i) j}(\sigma_n v,\sigma_n x,y)v_{i j}\bigg]-\underset{(v,x,y)}{\mathbb{E}}\bigg[\sum\limits_{i=1}^n{p}_i(v,x,y)\bigg]\\
		=&\frac{1}{n!}\sum\limits_{\sigma_n\in S_n}\underset{(v,x,y)}{\mathbb{E}}\bigg[\sum\limits_{i=1}^n\sum\limits_{j=1}^m{g}_{\sigma_n^{-1}(i) j}(\sigma_n v,\sigma_n x,y)[\sigma_n v]_{\sigma_n^{-1}(i)j}\bigg]-\underset{(v,x,y)}{\mathbb{E}}\bigg[\sum\limits_{i=1}^n{p}_i(v,x,y)\bigg]\\
		=&\frac{1}{n!}\sum\limits_{\sigma_n\in S_n}\underset{(v,x,y)}{\mathbb{E}}\bigg[\sum\limits_{i=1}^n\sum\limits_{j=1}^m{g}_{i j}(\sigma_n v,\sigma_n x,y)[\sigma_n v]_{i j}\bigg]-\underset{(v,x,y)}{\mathbb{E}}\bigg[\sum\limits_{i=1}^n{p}_i(v,x,y)\bigg]\\
		=&\frac{1}{n!}\sum\limits_{\sigma_n\in S_n}\underset{(v,x,y)}{\mathbb{E}}\bigg[\sum\limits_{i=1}^n\sum\limits_{j=1}^m{g}_{i j}(v,x,y)v_{i j}\bigg]-\underset{(v,x,y)}{\mathbb{E}}\bigg[\sum\limits_{i=1}^n{p}_i(v,x,y)\bigg]\\
		=&\underset{(v,x,y)}{\mathbb{E}}\bigg[\sum\limits_{i=1}^n\sum\limits_{j=1}^m{g}_{i j}(v,x,y)v_{i j}\bigg]-\underset{(v,x,y)}{\mathbb{E}}\bigg[\sum\limits_{i=1}^n{p}_i(v,x,y)\bigg]\\
		=&\underset{(v,x,y)}{\mathbb{E}}\bigg[\sum\limits_{i=1}^n u_i(v_i,v,x,y)\bigg].
	\end{align*}
	
Thus, we have completed the second step.

\paragraph{Step 3:} We lastly prove that the auction mechanism has a smaller ex-post regret after projection, {\it i.e.},
		\begin{equation*}\mathbb{E}\bigg[\max\limits_{v'\in\mathcal{V}^{n\times m}}\sum\limits_{i=1}^n[\mathcal{Q}_1{u}]_i(v_i,(v_i',v_{-i}),x,y)\bigg]\le\mathbb{E}\bigg[\max\limits_{v'\in\mathcal{V}^{n\times m}}\sum\limits_{i=1}^n{u}(v_i,(v_i',v_{-i}),x,y)\bigg].
		\end{equation*}
		
For the simplicity, we denote $\sum\limits_{j=1}^m g_{i j}v_{i j}$ by $\langle g_i,v_i\rangle$. Then, we have,
\begin{align*}
     &\underset{(v,x,y)}{\mathbb{E}}\bigg[\max\limits_{v'\in\mathcal{V}^{n\times m}}\sum\limits_{i=1}^n [\mathcal{Q}_1{u}]_i\big(v_i,(v_i',v_{-i}),x,y\big)\bigg]\\
     =&\underset{(v,x,y)}{\mathbb{E}}\bigg[\max\limits_{v'\in\mathcal{V}^{n\times m}}\sum\limits_{i=1}^n\sum\limits_{j=1}^m[\mathcal{Q}_1{g}]_{i j}\big((v'_i,v_{-i}),x,y\big)v_{i j}-[\mathcal{Q}_1{p}]_{i}\big((v_i',v_{-i}),x,y\big)\bigg]\\
     =&\underset{(v,x,y)}{\mathbb{E}}\bigg[\max\limits_{v'\in\mathcal{V}^{n\times m}}\sum\limits_{i=1}^n\big\langle[\mathcal{Q}_1{g}]_{i}\big((v'_i,v_{-i}),x,y\big),v_{i}\big\rangle-[\mathcal{Q}_1{p}]_{i}\big((v_i',v_{-i}),x,y\big)\bigg]\\
     =&\underset{(v,x,y)}{\mathbb{E}}\bigg[\max\limits_{v'\in\mathcal{V}^{n\times m}}\sum\limits_{i=1}^n\frac{1}{n!}\sum\limits_{\sigma_n\in S_n}\big\langle {g}_{\sigma_n^{-1}(i)}\big(\sigma_n(v'_i,v_{-i}),\sigma_n x,y\big),v_{i}\big\rangle-{p}_{\sigma_n^{-1}(i)}\big(\sigma_n(v_i',v_{-i}),\sigma_n x,y\big)\bigg].
\end{align*}  

Combining the following inequality, 
\begin{equation*}
    \max_{z}\sum_{k=1}^K f_k(z)\le \sum_{k=1}^K \max_{z} f_k(z),
\end{equation*}
we have that  
\begin{align*}
	&\underset{(v,x,y)}{\mathbb{E}}\bigg[\max\limits_{v'\in\mathcal{V}^{n\times m}}\sum\limits_{i=1}^n\frac{1}{n!}\sum\limits_{\sigma_n\in S_n}\big\langle {g}_{\sigma_n^{-1}(i)}\big(\sigma_n(v'_i,v_{-i}),\sigma_n x,y\big),v_{i}\big\rangle-{p}_{\sigma_n^{-1}(i)}\big(\sigma_n(v_i',v_{-i}),\sigma_n x,y\big)\bigg]\\
     \le&\underset{(v,x,y)}{\mathbb{E}}\bigg[\frac{1}{n!}\sum\limits_{\sigma_n\in S_n}\max\limits_{v'\in\mathcal{V}^{n\times m}}\sum\limits_{i=1}^n\big\langle {g}_{\sigma_n^{-1}(i)}\big(\sigma_n(v'_i,v_{-i}),\sigma_n x,y\big),v_{i}\big\rangle-{p}_{\sigma_n^{-1}(i)}\big(\sigma_n(v_i',v_{-i}),\sigma_n x,y\big)\bigg]\\
     =&\frac{1}{n!}\sum\limits_{\sigma_n\in S_n}\underset{(v,x,y)}{\mathbb{E}}\bigg[\max\limits_{v'\in\mathcal{V}^{n\times m}}\sum\limits_{i=1}^n\big\langle {g}_{\sigma_n^{-1}(i)}\big(\sigma_n(v'_i,v_{-i}),\sigma_n x,y\big),v_{i}\big\rangle-{p}_{\sigma_n^{-1}(i)}\big(\sigma_n(v_i',v_{-i}),\sigma_n x,y\big)\bigg]\\
     =&\frac{1}{n!}\sum\limits_{\sigma_n\in S_n}\underset{(v,x,y)}{\mathbb{E}}\bigg[\max\limits_{v'\in\mathcal{V}^{n\times m}}\sum\limits_{i=1}^n {u}_i\big(v_i,(v_i',v_{-i}),x,y\big)\bigg]\\
     =&\underset{(v,x,y)}{\mathbb{E}}\bigg[\max\limits_{v'\in\mathcal{V}^{n\times m}}\sum\limits_{i=1}^n {u}_i\big(v_i,(v_i',v_{-i}),x,y\big)\bigg].
\end{align*}

We thus have completed the proof of eqs. (\ref{equ1}) and (\ref{equ2}) when the orbit averaging $\mathcal{Q}_\cdot$ is the condition of bidder-symmetry.


Then, we prove this theorem in the condition of item-symmetry; {\it i.e.}, the orbit averaging $\mathcal{Q}_\cdot$ is the item averaging $\mathcal{Q}_2$, 
acting on the allocation rule $g$ and the payment rule $p$ as shown below,
\begin{equation*}
	\mathcal{Q}_2{g}(v,x,y)=\frac{1}{m!}\sum_{\sigma_m\in S_m}g(v \sigma_m,x,y\sigma_m)\sigma_m^{-1},
\end{equation*}
and,
\begin{equation*}
\mathcal{Q}_2{p}(v,x,y)=\frac{1}{m!}\sum_{\sigma_m\in S_m}p(v\sigma_m,x,y\sigma_m).
\end{equation*}

\paragraph{Step 1:} We first prove that the auction mechanism has the same expected revenue after projection, {\it i.e.}, \begin{equation*}
	\underset{(v,x,y)}{\mathbb{E}}\bigg[\sum\limits_{i=1}^n[\mathcal{Q}_2p]_i(v,x,y)\bigg]=\underset{(v,x,y)}{\mathbb{E}}\bigg[\sum\limits_{i=1}^n p_i(v,x,y)\bigg].
\end{equation*}

Since the valuation joint distribution is invariant under bidder permutation, we have 
    $\mathbb{E}_{(v,x,y)}[f(v\sigma_m,x,y\sigma_m)]=\mathbb{E}_{(v,x,y)}[f(v,x,y)].$
Then, we have
\begin{align*}
	&\underset{(v,x,y)}{\mathbb{E}}\bigg[\sum\limits_{i=1}^n [\mathcal{Q}_2{p}]_i(v,x,y)\bigg]\\
	=&\underset{(v,x,y)}{\mathbb{E}}\bigg[\frac{1}{m!}\sum\limits_{i=1}^n\sum\limits_{\sigma_m\in S_m}{p}_{i}(v\sigma_m,x,y\sigma_m)\bigg]\\
	=&\frac{1}{m!}\sum\limits_{i=1}^n\sum\limits_{\sigma_m\in S_m} \underset{(v,x,y)}{\mathbb{E}}\bigg[{p}_i(v\sigma_m,x,y\sigma_m)\bigg]\\
	=&\frac{1}{m!}\sum\limits_{i=1}^n\sum\limits_{\sigma_m\in S_m} \underset{(v,x,y)}{\mathbb{E}}\bigg[{p}_i(v,x,y)\bigg]
	=\sum\limits_{i=1}^n\underset{(v,x,y)}{\mathbb{E}}\bigg[{p}_i(v,x,y)\bigg].
\end{align*}

We have thus completed the first step.

\paragraph{Step 2:} We then prove that the sum of all bidders' utilities remains same after projection, {\it i.e.}, \begin{equation*}
	\underset{(v,x,y)}{\mathbb{E}}\bigg[\sum_{i=1}^n[\mathcal{Q}_2u]_i(v_i,v,x,y)\bigg]=\underset{(v,x,y)}{\mathbb{E}}\bigg[\sum_{i=1}^n u_i(v_i,v,x,y)\bigg].
\end{equation*}

Given that for all permutation $\pi$,
\begin{equation*}
    \sum_{j=1}^mg_{i\pi(j)}v_{i\pi(j)}=\sum_{j=1}^mg_{i j}v_{i j},
\end{equation*}
we have the following equation,
\begin{align*}
	&\underset{(v,x,y)}{\mathbb{E}}\bigg[\sum\limits_{i=1}^n[\mathcal{Q}_2{u}]_i(v_i,v,x,y)\bigg]\\
	=&\underset{(v,x,y)}{\mathbb{E}}\bigg[\sum\limits_{i=1}^n\Big[\sum\limits_{j=1}^m [\mathcal{Q}_2{g}]_{i j}(v,x,y)v_{i j}-[\mathcal{Q}_2{p}]_i(v,x,y)\Big]\bigg]\\
	=&\underset{(v,x,y)}{\mathbb{E}}\bigg[\sum\limits_{i=1}^n\sum\limits_{j=1}^m[\mathcal{Q}_2{g}]_{i j}(v,x,y)v_{i j} \bigg]-\underset{(v,x,y)}{\mathbb{E}}\bigg[\sum\limits_{i=1}^n{p}_i(v,x,y)\bigg]\\
	=&\underset{(v,x,y)}{\mathbb{E}}\bigg[\frac{1}{m!}\sum\limits_{i=1}^n\sum\limits_{j=1}^m\sum\limits_{\sigma_m\in S_m}{g}_{i \sigma_m^{-1}(j)}(v\sigma_m,x,y\sigma_m)v_{i j}\bigg]-\underset{(v,x,y)}{\mathbb{E}}\bigg[\sum\limits_{i=1}^n{p}_i(v,x,y)\bigg]\\
	=&\frac{1}{m!}\sum\limits_{\sigma_m\in S_m}\underset{(v,x,y)}{\mathbb{E}}\bigg[\sum\limits_{i=1}^n\sum\limits_{j=1}^m{g}_{i \sigma_m^{-1}(j)}(v\sigma_m,x,y\sigma_m)[v\sigma_m]_{i \sigma_m^{-1}(j)}\bigg]-\underset{(v,x,y)}{\mathbb{E}}\bigg[\sum\limits_{i=1}^n{p}_i(v,x,y)\bigg]\\
	=&\frac{1}{m!}\sum\limits_{\sigma_m\in S_m}\underset{(v,x,y)}{\mathbb{E}}\bigg[\sum\limits_{i=1}^n\sum\limits_{j=1}^m{g}_{i j}(v\sigma_m,x,y\sigma_m)[v\sigma_m]_{i j}\bigg]-\underset{(v,x,y)}{\mathbb{E}}\bigg[\sum\limits_{i=1}^n{p}_i(v,x,y)\bigg]\\
	=&\frac{1}{m!}\sum\limits_{\sigma_m\in S_m}\underset{(v,x,y)}{\mathbb{E}}\bigg[\sum\limits_{i=1}^n\sum\limits_{j=1}^m{g}_{i j}(v,x,y)v_{i j}\bigg]-\underset{(v,x,y)}{\mathbb{E}}\bigg[\sum\limits_{i=1}^n{p}_i(v,x,y)\bigg]\\
	=&\underset{(v,x,y)}{\mathbb{E}}\bigg[\sum\limits_{i=1}^n\sum\limits_{j=1}^m{g}_{i j}(v,x,y)v_{i j}\bigg]-\underset{(v,x,y)}{\mathbb{E}}\bigg[\sum\limits_{i=1}^n{p}_i(v,x,y)\bigg]\\
	=&\underset{(v,x,y)}{\mathbb{E}}\bigg[\sum\limits_{i=1}^n u_i(v_i,v,x,y)\bigg].
\end{align*}

Thus, we have completed the second step.

\paragraph{Step 3:}
We lastly prove that the auction mechanism have a smaller ex-post regret after projection, {\it i.e.},
\begin{equation*}\mathbb{E}\bigg[\max\limits_{v'\in\mathcal{V}^{n\times m}}\sum\limits_{i=1}^n[\mathcal{Q}_2{u}]_i(v_i,(v_i',v_{-i}),x,y)\bigg]\le\mathbb{E}\bigg[\max\limits_{v'\in\mathcal{V}^{n\times m}}\sum\limits_{i=1}^n{u}(v_i,(v_i',v_{-i}),x,y)\bigg].
\end{equation*}

By the definition of the utility $u$ and the item averaging $\mathcal{Q}_2$, we have
\begin{align*}
	&\underset{(v,x,y)}{\mathbb{E}}\bigg[\max\limits_{v'\in\mathcal{V}^{n\times m}}\sum\limits_{i=1}^n [\mathcal{Q}_2{u}]_i\big(v_i,(v_i',v_{-i}),x,y\big)\bigg]\\
	=&\underset{(v,x,y)}{\mathbb{E}}\bigg[\max\limits_{v'\in\mathcal{V}^{n\times m}}\sum\limits_{i=1}^n\sum\limits_{j=1}^m[\mathcal{Q}_2{g}]_{i j}\big((v'_i,v_{-i}),x,y\big)v_{i j}-[\mathcal{Q}_2{p}]_{i}\big((v_i',v_{-i}),x,y\big)\bigg]\\
	=&\underset{(v,x,y)}{\mathbb{E}}\bigg[\max\limits_{v'\in\mathcal{V}^{n\times m}}\sum\limits_{i=1}^n\big\langle[\mathcal{Q}_2{g}]_{i}\big((v'_i,v_{-i}),x,y\big),v_{i}\big\rangle-[\mathcal{Q}_2{p}]_{i}\big((v_i',v_{-i}),x,y\big)\bigg]\\
	=&\underset{(v,x,y)}{\mathbb{E}}\bigg[\max\limits_{v'\in\mathcal{V}^{n\times m}}\sum\limits_{i=1}^n\frac{1}{m!}\sum\limits_{\sigma_m\in S_m}\big\langle {g}_{i}\big((v'_i,v_{-i})\sigma_m,x,y\sigma_m\big)\sigma_m^{-1},v_{i}\big\rangle-{p}_{i}\big((v_i',v_{-i})\sigma_m,x,y\sigma_m\big)\bigg].
\end{align*}  

Combining the following inequality
\begin{equation*}
    \max_{z}\sum_{k=1}^K f_k(z)\le \sum_{k=1}^K \max_{z} f_k(z),
\end{equation*}
we have  
\begin{align*}
	&\underset{(v,x,y)}{\mathbb{E}}\bigg[\max\limits_{v'\in\mathcal{V}^{n\times m}}\sum\limits_{i=1}^n\frac{1}{m!}\sum\limits_{\sigma_m\in S_m}\big\langle {g}_{i}\big((v'_i,v_{-i})\sigma_m,x,y\sigma_m\big)\sigma_m^{-1},v_{i}\big\rangle-{p}_{i}\big((v_i',v_{-i})\sigma_m,x,y\sigma_m\big)\bigg]\\
	\le&\underset{(v,x,y)}{\mathbb{E}}\bigg[\frac{1}{m!}\sum\limits_{\sigma_m\in S_m}\sum\limits_{i=1}^n\max\limits_{v'\in\mathcal{V}^{n\times m}}\big\langle {g}_{i}\big((v'_i,v_{-i})\sigma_m,x,y\sigma_m\big)\sigma_m^{-1},v_{i}\big\rangle-{p}_{i}\big((v_i',v_{-i})\sigma_m,x,y\sigma_m\big)\bigg]\\
	=&\frac{1}{m!}\sum\limits_{\sigma_m\in S_m}\underset{(v,x,y)}{\mathbb{E}}\bigg[\max\limits_{v'\in\mathcal{V}^{n\times m}}\sum\limits_{i=1}^n\big\langle {g}_{i}\big((v'_i,v_{-i})\sigma_m,x,y\sigma_m\big),v_{i}\sigma_m\big\rangle-{p}_{i}\big((v_i',v_{-i})\sigma_m,x,y\sigma_m\big)\bigg]\\
	=&\frac{1}{m!}\sum\limits_{\sigma_m\in S_m}\underset{(v,x,y)}{\mathbb{E}}\bigg[\max\limits_{v'\in\mathcal{V}^{n\times m}}\sum\limits_{i=1}^n {u}_i\big(v_i,(v_i',v_{-i}),x,y\big)\bigg]\\
	=&\underset{(v,x,y)}{\mathbb{E}}\bigg[\max\limits_{v'\in\mathcal{V}^{n\times m}}\sum\limits_{i=1}^n {u}_i\big(v_i,(v_i',v_{-i}),x,y\big)\bigg].
\end{align*}

{We thus have proved this theorem in the condition of item-symmetry.}

The proofs are completed.
\end{proof}

\subsection{Proof of Lemma \ref{lemma:composition}}

In this section, we present the proof of Lemma \ref{lemma:composition}.

\begin{proof}[Proof of Lemma \ref{lemma:composition}.]
Both the bidder averaging and the item averaging are linear. Thus, we have the following results,
	\begin{align*}
		&\mathcal{Q}_1\circ\mathcal{Q}_2 f(v,x,y)\\
		=&\mathcal{Q}_1 \bigg[\frac{1}{m!}\sum_{\sigma_m\in S_m}f(v\sigma_m,x,y\sigma_m)\sigma_m^{-1}\bigg]\\
		=&\frac{1}{n!}\sum_{\sigma_n\in S_n}\sigma_n^{-1}\bigg[\frac{1}{m!}\sum_{\sigma_m\in S_m}f(\sigma_n v\sigma_m,\sigma_n x,y\sigma_m)\sigma_m^{-1}\bigg]\\
		=&\frac{1}{n!m!}\sum_{\sigma_n\in S_n}\sum_{\sigma_m\in S_m}\sigma_n^{-1}f(\sigma_n v\sigma_m,\sigma_n x,y\sigma_m)\sigma_m^{-1}=\mathcal{Q}_3f(v,x,y),
	\end{align*}
and
	\begin{align*}
		&\mathcal{Q}_2\circ \mathcal{Q}_1 f(v,x,y)\\
		=&\mathcal{Q}_2\bigg[\frac{1}{n!}\sum_{\sigma_n\in S_n}\sigma_n^{-1}f(\sigma_n v,\sigma_n x,y)\bigg]\\
		=&\frac{1}{m!}\sum_{\sigma_m\in S_m}\bigg[\frac{1}{n!}\sum_{\sigma_n\in S_n}\sigma_n^{-1}f(\sigma_n v\sigma_m,\sigma_n x,y\sigma_m)\bigg]\sigma_m^{-1}\\
		=&\frac{1}{m!n!}\sum_{\sigma_n\in S_n}\sum_{\sigma_m\in S_m}\sigma_n^{-1}f(\sigma_n v\sigma_m,\sigma_n x,y\sigma_m)\sigma_m^{-1}=\mathcal{Q}_3f(v,x,y).
	\end{align*}

The above two equations hold for any $f$. Then, we may prove that
\begin{equation*}
    \mathcal{Q}_3=\mathcal{Q}_1\circ \mathcal{Q}_2=\mathcal{Q}_2\circ\mathcal{Q}_1,
\end{equation*}
which is exactly the claim of this theorem.

The proof is completed.
\end{proof}

\subsection{Proof of Theorem \ref{theorem plus}}

In this section, we apply our Lemma \ref{lemma:composition} and Theorem \ref{the} to prove Theorem \ref{theorem plus}.

	\begin{proof}[Proof of Theorem \ref{theorem plus}]
		
		For the simplicity, we rewrite $\mathcal{Q}_3p$ and $\mathcal{Q}_3reg$ as $\mathcal{Q}_2(\mathcal{Q}_1p)$ and $\mathcal{Q}_2(\mathcal{Q}_1reg)$, respectively. Then, for a payment rule $p$, we have that,
		\begin{align*}
			&\underset{(v,x,y)}{\mathbb{E}}\bigg[\sum_{i=1}^n[\mathcal{Q}_3p]_i(v,x,y)\bigg]\\
			=&\underset{(v,x,y)}{\mathbb{E}}\bigg[\sum_{i=1}^n[\mathcal{Q}_2(\mathcal{Q}_1p)]_i(v,x,y)\bigg]
			=\underset{(v,x,y)}{\mathbb{E}}\bigg[\sum_{i=1}^n[\mathcal{Q}_1p]_i(v,x,y)\bigg]\\
			=&\underset{(v,x,y)}{\mathbb{E}}\bigg[\sum_{i=1}^n p_i(v,x,y)\bigg].
		\end{align*}
		
	Also, we have the following result,
	\begin{align*}
		&\underset{(v,x,y)}{\mathbb{E}}\bigg[\sum_{i=1}^n [\mathcal{Q}_3reg]_i(v,x,y)\bigg]\\
		=&\underset{(v,x,y)}{\mathbb{E}}\bigg[\sum_{i=1}^n [\mathcal{Q}_2(\mathcal{Q}_1reg)]_i(v,x,y)\bigg]
		\le \underset{(v,x,y)}{\mathbb{E}}\bigg[\sum_{i=1}^n [\mathcal{Q}_1reg]_i(v,x,y)\bigg]\\
		\le&\underset{(v,x,y)}{\mathbb{E}}\bigg[\sum_{i=1}^n reg_i(v,x,y)\bigg].
	\end{align*}
	
	Moreover, we have the following result on the regret gap $\Delta_3$,
	\begin{equation*}
	    \underset{(v,x,y)}{\mathbb{E}}[\Delta_3(g,p;v,x,y)]=\underset{(v,x,y)}{\mathbb{E}}[\Delta_1(g,p;v,x,y)]+\underset{(v,x,y)}{\mathbb{E}}[\Delta_2(\mathcal{Q}_1g,\mathcal{Q}_1p;v,x,y)]\ge 0.
	\end{equation*}
	
This proof is completed.
	\end{proof}

\subsection{Proof of Theorem \ref{covering}}
In this section, we present the proof of Theorem \ref{covering}. 

We start with the definitions or notations necessary for our proof. We define the allocation rule space and the payment rule space as follows, 
\begin{equation*}
	\mathcal{G}=\{g^\omega:\omega\in\Omega\}~~\text{and}~~\mathcal{P}=\{p^\omega:\omega\in\Omega\},
\end{equation*}
where $\omega$ is the auction mechanism parameter and $\Omega$ is the set of all feasible parameters.
We then define the induced utility and ex-post regret spaces as follows,
\begin{equation*}
	\mathcal{U}=\Big\{u^\omega:u_i^\omega(v_i',v,x,y)=\sum\limits_{j=1}^m g^\omega_{i j} (v,x,y)v_{i j}'-p^\omega_i(v,x,y)\Big\},
\end{equation*}
and
\begin{equation*}
	\mathcal{R}=\Big\{reg^\omega:reg^\omega_i(v,x,y)=\max_{v_i'}u^{\omega}(v_i,(v_i',v_{-i}),x,y)-u^{\omega}(v_i,v,x,y)\Big\}.
\end{equation*}
Then, the $l_{\infty,1}$-distance on $\mathcal{U}$ and $\mathcal{P}$ is defined as below,
\begin{equation*}
	l_{\infty,1}(u,u')=\max\limits_{(v,v_i',x,y)}\bigg(\sum\limits_{i=1}^n|u_i(v_i,(v_i',v_{-i}),x,y)-u_i'(v_i,(v_i',v_{-i}),x,y)|\bigg),
\end{equation*}
and
\begin{equation*} l_{\infty,1}(p,p')=\max\limits_{(v,x,y)}\bigg(\sum\limits_{i=1}^n|p_i(v,x,y)-p_i'(v,x,y)|\bigg).
\end{equation*}

We now present the proof of Theorem \ref{covering}.
\begin{proof}[Proof of Theorem \ref{covering}.]
We prove Theorem \ref{covering} in two steps:
(1) we first prove that the distance between any two mechanisms is smaller when we project them to be permutation-equivariant; (2) then, we prove that the smaller distance implies a smaller covering number. 
 
 \paragraph{Step 1:} We prove that the distance between two mechanisms becomes smaller after projection, 
 {\it i.e.},
\begin{equation*}
     l_{\infty,1}(\mathcal{Q}_\cdot{p},\mathcal{Q}_\cdot{p'})\le l_{\infty,1}(p,p'),
\end{equation*}
and
\begin{equation*}
l_{\infty,1}(\mathcal{Q}_\cdot{u},\mathcal{Q}_\cdot{u'})\le l_{\infty,1}(u,u'),
\end{equation*}
where $u,u'\in\mathcal{U}$, $p,p' \in \mathcal{P}$, and $\mathcal{Q}_\cdot=\mathcal{Q}_1$ or $\mathcal{Q}_2$.

When $\mathcal{Q}_\cdot$ is $\mathcal{Q}_1$, we have that
\begin{align*}
    &l_{\infty,1}(\mathcal{Q}_1{p},\mathcal{Q}_1{p}')\\
    =&\max\limits_{(v,x,y)}\sum\limits_{i=1}^n|\mathcal{Q}_1{p}_i(v,x,y)-\mathcal{Q}_1{p}_i'(v,x,y)|\\
    =&\max\limits_{(v,x,y)}\sum\limits_{i=1}^n\bigg|\frac{1}{n!}\sum\limits_{\sigma_n\in S_n}\big[{p}_{\sigma_n^{-1}(i)}(\sigma_n v,\sigma_n x,y)-{p}_{\sigma_n^{-1}(i)}'(\sigma_n v,\sigma_n x,y)\big]\bigg|\\
    \le&\max\limits_{(v,x,y)}\sum\limits_{i=1}^n\frac{1}{n!}\sum\limits_{\sigma_n\in S_n}\big|{p}_{\sigma_n^{-1}(i)}(\sigma_n v,\sigma_n x,y)-{p}_{\sigma_n^{-1}(i)}'(\sigma_n v,\sigma_n x,y)\big|\\
    \le&\sum\limits_{\sigma_n\in S_n}\frac{1}{n!}\max\limits_{(v,x,y)}\sum\limits_{i=1}^n \big|{p}_{\sigma_n^{-1}(i)}(\sigma_n v,\sigma_n x,y)-{p}_{\sigma_n^{-1}(i)}'(\sigma_n v,\sigma_n x,y)\big|\\
    =&\sum\limits_{\sigma_n\in S_n}\frac{1}{n!}\max\limits_{(v,x,y)}\sum\limits_{i=1}^n \big|{p}_{i}(\sigma_n v,\sigma_n x,y)-{p}_{i}'(\sigma_n v,\sigma_n x,y)\big|\\
    =&\sum\limits_{\sigma_n\in S_n}\frac{1}{n!}\max\limits_{(v,x,y)}\sum\limits_{i=1}^n\big|{p}_{i}(v,x,y)-{p}_{i}'(v,x,y)\big|\\
    =&\max\limits_{(v,x,y)}\sum\limits_{i=1}^n\big|{p}_{i}(v,x,y)-{p}_{i}'(v,x,y)\big|\\
    =&l_{\infty,1}({p},{p}'),
\end{align*} 
and 
\begin{align*}
    &l_{\infty,1}(\mathcal{Q}_1{u},\mathcal{Q}_1{u}')\\
    =&\max\limits_{v,v',x,y}\sum\limits_{i=1}^n|[\mathcal{Q}_1{u}]_i(v_i,(v_i',v_{-i}),x,y)-[\mathcal{Q}_1{u'}]_i(v_i,(v_i',v_{-i}),x,y)|\\
    =&\max\limits_{v,v',x,y}\sum\limits_{i=1}^n\bigg|\frac{1}{n!}\sum\limits_{\sigma_n\in S_n}{u}_{\sigma_n^{-1}(i)}(v_i,\sigma_{n}(v_i',v_{-i}),\sigma_n x,y)-{u}_{\sigma_n^{-1}(i)}'(v_i,\sigma_n (v_i',v_{-i}),\sigma_n x,y)\bigg|\\
    \le&\max\limits_{v,v',x,y}\sum\limits_{i=1}^n\frac{1}{n!}\sum\limits_{\sigma_n\in S_n}|{u}_{\sigma_n^{-1}(i)}(v_i,\sigma_{n}(v_i',v_{-i}),\sigma_n x,y)-{u}_{\sigma_n^{-1}(i)}'(v_i,\sigma_n (v_i',v_{-i}),\sigma_n x,y)|\\
    \le &\frac{1}{n!}\sum\limits_{\sigma_n\in S_n}\max\limits_{v,v',x,y}\sum\limits_{i=1}^n|{u}_{\sigma_n^{-1}(i)}( v_i,\sigma_{n}(v_i',v_{-i}),\sigma_n x,y)-{u}_{\sigma_n^{-1}(i)}'( v_i,\sigma_n (v_i',v_{-i}),\sigma_n x,y)|\\
    =&\frac{1}{n!}\sum\limits_{\sigma_n\in S_n}\max\limits_{v,v',x,y}\sum\limits_{i=1}^n|{u}_{i}(v_i,(v_i',v_{-i}),x,y)-{u}_{i}'(v_i,(v_i',v_{-i}),x,y)|\\
    =&\max\limits_{v,v',x,y}\sum\limits_{i=1}^n|{u}_{i}(v_i,(v_i',v_{-i}),x,y)-{u}_{i}'(v_i,(v_i',v_{-i}),x,y)|\\
    =&l_{\infty,1}({u},{u}').
\end{align*} 

Then, when $\mathcal{Q}_\cdot$ is $\mathcal{Q}_2$, we prove the result as below,
\begin{align*}
	&l_{\infty,1}(\mathcal{Q}_2{p},\mathcal{Q}_2{p}')\\
	=&\max\limits_{(v,x,y)}\sum\limits_{i=1}^n|\mathcal{Q}_2{p}_i(v,x,y)-\mathcal{Q}_2{p}_i'(v,x,y)|\\
	=&\max\limits_{(v,x,y)}\sum\limits_{i=1}^n\bigg|\frac{1}{m!}\sum\limits_{\sigma_m\in S_m}\big[{p}_{i}(v\sigma_m,x,y\sigma_m)-{p}_{i}'(v\sigma_m,x,y\sigma_m)\big]\bigg|\\
	\le&\max\limits_{(v,x,y)}\sum\limits_{i=1}^n\frac{1}{m!}\sum\limits_{\sigma_m\in S_m}\big|{p}_{i}(v\sigma_m,x,y\sigma_m)-{p}_{i}'(v\sigma_m,x,y\sigma_m)\big|\\
	\le&\sum\limits_{\sigma_m\in S_m}\frac{1}{m!}\max\limits_{(v,x,y)}\sum\limits_{i=1}^n \big|{p}_{i}(v\sigma_m,x,y\sigma_m)-{p}_{i}'(v\sigma_m,x,y\sigma_m)\big|\\
	=&\sum\limits_{\sigma_m\in S_m}\frac{1}{m!}\max\limits_{(v,x,y)}\sum\limits_{i=1}^n\big|{p}_{i}(v,x,y)-{p}_{i}'(v,x,y)\big|\\
	=&\max\limits_{(v,x,y)}\sum\limits_{i=1}^n\big|{p}_{i}(v,x,y)-{p}_{i}'(v,x,y)\big|\\
	=&l_{\infty,1}({p},{p}'),
\end{align*} 
and 
\begin{align*}
	&l_{\infty,1}(\mathcal{Q}_2{u},\mathcal{Q}_2{u'})\\
	=&\max\limits_{v,v',x,y}\sum\limits_{i=1}^n|[\mathcal{Q}_2{u}]_i(v_i,(v_i',v_{-i}),x,y)-[\mathcal{Q}_2{u'}]_i(v_i,(v_i',v_{-i}),x,y)|\\
	=&\max\limits_{v,v',x,y}\sum\limits_{i=1}^n\bigg|\frac{1}{m!}\sum\limits_{\sigma_m\in S_m}|{u}_{i}(v_i\sigma_m,(v_i',v_{-i})\sigma_m,x,y\sigma_m)-{u}_{i}'(v_i\sigma_m,(v_i',v_{-i})\sigma_m,x,y\sigma_m)\bigg|\\
	\le&\max\limits_{v,v',x,y}\sum\limits_{i=1}^n\frac{1}{m!}\sum\limits_{\sigma_m\in S_m}|{u}_{i}(v_i\sigma_m,(v_i',v_{-i})\sigma_m,x,y\sigma_m)-{u}_{i}'(v_i\sigma_m,(v_i',v_{-i})\sigma_m,x,y\sigma_m)|\\
	\le &\frac{1}{m!}\sum\limits_{\sigma_m\in S_m}\max\limits_{v,v',x,y}\sum\limits_{i=1}^n|{u}_{i}(v_i\sigma_m,(v_i',v_{-i})\sigma_m,x,y\sigma_m)-{u}_{i}'(v_i\sigma_m,(v_i',v_{-i})\sigma_m,x,y\sigma_m)|\\
	=&\frac{1}{m!}\sum\limits_{\sigma_m\in S_m}\max\limits_{v,v',x,y}\sum\limits_{i=1}^n|{u}_{i}(v_i,(v_i',v_{-i}),x,y)-{u}_{i}'(v_i,(v_i',v_{-i}),x,y)|\\
	=&\max\limits_{v,v',x,y}\sum\limits_{i=1}^n|{u}_{i}(v_i,(v_i',v_{-i}),x,y)-{u}_{i}'(v_i,(v_i',v_{-i}),x,y)|\\
	=&l_{\infty,1}({u},{u}').
\end{align*} 

Thus, we have completed Step 1.

\paragraph{Step 2:} We prove that a smaller distance implies a smaller covering number.

Let $\mathcal{X}$ and $\mathcal{Y}$ be two metric spaces with two different distances $l_1$ and $l_2$, respectively. There exists a surjective mapping $f$ from $\mathcal{Y}$ to $\mathcal{X}$, such that $l_1(f(x),f(y))\le l_2(x,y)$ for all $x,y\in \mathcal{Y}$. The covering numbers $\mathcal{N}_1(\mathcal{X},r)$ and $\mathcal{N}_2(\mathcal{Y},r)$ are defined as the minimum numbers of balls with radius r that can cover $\mathcal{X}$ and $\mathcal{Y}$ under $l_1$ and $l_2$, respectively.

By the definition of the covering number $\mathcal{N}_2(\mathcal{Y},r)$, there exists a set $\mathcal{A}$ of scale $\mathcal{N}_2(\mathcal{Y},r)$, such that 
\begin{equation*}
	l_2(x,\mathcal{A})=\inf\limits_{y\in\mathcal{A}}l_2(x,y)<r, \forall x\in\mathcal{Y}.
\end{equation*}

Then, $f(\mathcal{A})$ is also a $r$-cover for $\mathcal{X}$ under distance $l_1$, {\it i.e.}, for any $x\in\mathcal{Y}$, we have
\begin{equation*}
	l_1(f(x),f(\mathcal{A}))=\inf\limits_{y\in\mathcal{A}}l_1(f(x),f(y))\le\inf\limits_{y\in\mathcal{A}}l_2(x,y)=l_2(x,\mathcal{A})<r.
\end{equation*}

Because $f$ is surjective, for any $x'\in\mathcal{X}$, there exists an $x\in\mathcal{Y}$, such that $x'=f(x)$. Then, for any $x'\in\mathcal{X}$, we have that
\begin{equation*}
	l_1(x',f(\mathcal{A}))=	l_1(f(x),f(\mathcal{A}))<r.
\end{equation*}

By the definition of $\mathcal{N}_1(\mathcal{X},r)$, we have
\begin{equation*}
	\mathcal{N}_1(\mathcal{X},r)\le |f(\mathcal{A})|\le |\mathcal{A}|=\mathcal{N}_2(\mathcal{Y},r).
\end{equation*}
 
Eventually, combining the results in Step 1 and in Step 2, we have that
\begin{equation*}
    \mathcal{N}_{\infty,1}(\mathcal{Q}_\cdot\mathcal{U},r)\le\mathcal{N}_{\infty,1}(\mathcal{U},r),
\end{equation*}
and 
\begin{equation*}
    \mathcal{N}_{\infty,1}(\mathcal{Q}_\cdot\mathcal{P},r)\le\mathcal{N}_{\infty,1}(\mathcal{P},r),
\end{equation*}
for both the bidder averaging $\mathcal{Q}_1$ and the item averaging $\mathcal{Q}_2$.
\end{proof}

\subsection{Proof of Theorem \ref{gen} and Corollary \ref{corollary}}

We first introduce two lemmas. The first lemma gives a concentration inequality via the covering number. This result can be used to bound the gap between expected revenue/ex-post regret and empirical revenue/ex-post regret. The second lemma bounds the covering number $\mathcal{N}_{\infty,1}(\mathcal{R},2r)$ by the covering number $\mathcal{N}_{\infty,1}(\mathcal{U},r)$. Both lemmas has been proved by \cite{duan2022context}. We recall them here to make our paper completed.

\begin{lemma}[cf. Lemma E.1, \cite{duan2022context}]\label{lemma}
	Let $\mathcal{S}=\{z_1,\dots,z_L\}$ be a set of {i.i.d.} sample points drawn from a distribution $\mathcal{D}$ over $\mathcal{Z}$. Suppose $\mathcal{F}$ is a set of functions from $\mathcal{Z}$ to $\mathbb{R}$ such that $f(z)\in [a,b]$ for all $f\in\mathcal{F}$ and $z\in\mathcal{Z}$. We define $l_{\infty}$ on $\mathcal{F}$ as
	\begin{equation*}
	    l_\infty(f,f')=\max_{z\in\mathcal{Z}}|f(z)-f'(z)|,
	\end{equation*}
	and $\mathcal{N}_{\infty}(\mathcal{F},r)$ as the minimum number of balls with radius $r$ that can cover $\mathcal{F}$ under $l_\infty$-distance. Then, we have the following concentration inequality, 
	\begin{equation*}
		\mathbb{P}\bigg[\exists f\in\mathcal{F}:\Big|\frac{1}{L}\sum\limits_{i=1}^L f(z_i)-\mathbb{E}[f(z)]\Big|>\epsilon\bigg]\le2\mathcal{N}_{\infty}\Big(\mathcal{F},\frac{\epsilon}{3}\Big)\exp\Big(-\frac{2L\epsilon^2}{9(b-a)^2}\Big).
	\end{equation*}
\end{lemma}	

\begin{proof}
	By the definition of $\mathcal{N}_{\infty}(\mathcal{F},r)$, for any $f\in\mathcal{F}$, there exists an $f_r\in \mathcal{F}_r$ such that $\mathcal{F}_r$ is an $r$-cover for $\mathcal{F}$ and $l_{\infty}(f,f_r)<r$. Denote $\frac{1}{L}\sum_{i=1}^L f(z_i)$ by $\mathbb{E}_{\mathcal{S}}[f(z)]$. Then, we have
	\begin{align*}
		&\mathbb{P}\bigg[\exists f\in\mathcal{F}:\Big|\mathbb{E}_{\mathcal{S}}[f(z)]-\mathbb{E}[f(z)]\Big|>\epsilon\bigg]\\
		=&	\mathbb{P}\bigg[\exists f\in\mathcal{F}:\Big|\mathbb{E}_{\mathcal{S}}[f(z)]-\mathbb{E}_{\mathcal{S}}[f_r(z)]+\mathbb{E}_{\mathcal{S}}[f_r(z)]-\mathbb{E}[f_r(z)]+\mathbb{E}[f_r(z)]-\mathbb{E}[f(z)]\Big|>\epsilon\bigg]\\
		\le&\mathbb{P}\bigg[\exists f\in\mathcal{F}:\Big|\mathbb{E}_{\mathcal{S}}[f(z)]-\mathbb{E}_{\mathcal{S}}[f_r(z)]\Big|+\Big|\mathbb{E}_{\mathcal{S}}[f_r(z)]-\mathbb{E}[f_r(z)]\Big|+\Big|\mathbb{E}[f_r(z)]-\mathbb{E}[f(z)]\Big|>\epsilon\bigg]\\
		\le&\mathbb{P}\bigg[\exists f_{r}\in\mathcal{F}_{\frac{\epsilon}{3}}:\Big|\mathbb{E}_{\mathcal{S}}[f_r(z)]-\mathbb{E}[f_r(z)]\Big|>\frac{\epsilon}{3}\bigg]\\
		\le&\mathcal{N}_{\infty} \Big(\mathcal{F},\frac{\epsilon}{3}\Big)\mathbb{P}\bigg[\Big|\mathbb{E}_{\mathcal{S}}[f(z)]-\mathbb{E}[f(z)]\Big|>\frac{\epsilon}{3}\bigg]\\
		\le &2\mathcal{N}_{\infty} \Big(\mathcal{F},\frac{\epsilon}{3}\Big)\exp\Big(-\frac{2L\epsilon^2}{9(b-a)^2}\Big).
	\end{align*}
	
	The third inequality follows from the fact that when $r=\frac{\epsilon}{3}$, we have
	\begin{equation*}
	 |f(z)-f_r(z)|<\frac{\epsilon}{3},   
	\end{equation*}
	for all $z\in\mathcal{Z}$ and $f\in\mathcal{F}$. Then, from the Hoeffding's inequality, we have
	\begin{equation*}
		\Big|\frac{1}{L}\sum\limits_{i=1}^Lf(z_i)-\frac{1}{L}\sum\limits_{i=1}^Lf_r(z_i)\Big|<\frac{\epsilon}{3}~~\text{and}~~\Big|\mathbb{E}[f(z)]-\mathbb{E}[f_r(z)]\Big|<\frac{\epsilon}{3}.
	\end{equation*}
	
	The proof is completed.
\end{proof}

The following lemma bounds the covering number $\mathcal{N}_{\infty,1}(\mathcal{R},2r)$ by the covering number $\mathcal{N}_{\infty,1}(\mathcal{U},r)$. Then the gap between expected ex-post regret and empirical ex-post regret can be bounded by the covering number $\mathcal{N}_{\infty,1}(\mathcal{U},r)$.

\begin{lemma}[cf. Lemma E.3, \cite{duan2022context}]\label{lemma2}
	We define $\mathcal{N}_{\infty,1}(\mathcal{R},r)$ and $\mathcal{N}_{\infty,1}(\mathcal{U},r)$ as the minimum numbers of balls with radius $r$ that can cover spaces $\mathcal{R}$ and $\mathcal{U}$ under distance $l_{\infty,1}$, respectively. Then, we have that
	\begin{equation*}
		\mathcal{N}_{\infty,1}(\mathcal{R},2r)\le \mathcal{N}_{\infty,1}(\mathcal{U},r)
	\end{equation*}
\end{lemma}

\begin{proof}
	By the definition of $\mathcal{N}_{\infty,1}(\mathcal{U},r)$, there exists an $r$-cover $\mathcal{U}_r$ for $\mathcal{U}$, such that $|\mathcal{U}_r|=\mathcal{N}_{\infty,1}(\mathcal{U},r)$ and for any $u\in\mathcal{U}$, 
	\begin{equation*}
		l_{\infty,1}(u,\mathcal{U}_r)=\inf_{u'\in\mathcal{U}_r}l_{\infty,1}(u,u')<r.
	\end{equation*}
	We define $\mathcal{R}_r$ as 
	\begin{equation*}
		\{reg\in\mathcal{R}:reg_i(v,x,y)=\max_{v_i'}u_i(v_i,(v_i',v_{-i}),x,y)-u_i(v_i,v,x,y) \text{ for some }u_i\in\mathcal{U}_r\}.
	\end{equation*}
	Then, we can prove that $\mathcal{R}_r$ is a $2r$-cover for the space $\mathcal{R}$, {\it i.e.},
	\begin{align*}
		&l_{\infty,1}(reg,\mathcal{R}_r)\\
		=&\inf_{reg'\in\mathcal{R}_r}l_{\infty,1}(reg,reg')\\
		=&\inf_{reg'\in\mathcal{R}_r}\max_{v,x,y}\sum_{i=1}^n|reg_i(v,x,y)-reg_i'(v,x,y)|\\
		=&\inf_{u'\in\mathcal{U}_r}\max_{v,x,y}\sum_{i=1}^n\Big|[\max_{v_i'}u_i(v_i,(v_i',v_{-i}),x,y)-u_i(v_i,v,x,y)]\\
		&-[\max_{v_i'}u'_i(v_i,(v_i',v_{-i}),x,y)-u'_i(v_i,v,x,y)]\Big|\\
		\le&\inf_{u'\in\mathcal{U}_r}\max_{v,x,y}\bigg[\sum_{i=1}^n\Big|\max_{v_i'}u_i(v_i,(v_i',v_{-i}),x,y)
		-\max_{v_i'}u'_i(v_i,(v_i',v_{-i}),x,y)\Big|\\
		&+\Big|u_i(v_i,v,x,y)-u_i'(v_i,v,x,y)\Big|\bigg]\\
		\le&\max_{v,x,y}\bigg[\sum_{i=1}^n\Big|\max_{v_i'}u_i(v_i,(v_i',v_{-i}),x,y)
		-\max_{v_i'}u^*_i(v_i,(v_i',v_{-i}),x,y)\Big|\\
		&+\Big|u_i(v_i,v,x,y)-u^*_i(v_i,v,x,y)\Big|\bigg]~~\text{(where $l_{\infty,1}(u,u^*)<r$)}\\
		\le&\max_{v,x,y}\sum_{i=1}^n\Big|\max_{v_i'}u_i(v_i,(v_i',v_{-i}),x,y)
		-\max_{v_i'}u^*_i(v_i,(v_i',v_{-i}),x,y)\Big|+r\\
		\le&\max_{v,x,y}\sum_{i=1}^n\max_{v_i'}|u_i(v_i,(v_i',v_{-i}),x,y)-u^*_i(v_i,(v_i',v_{-i}),x,y)|+r\\
		=&\max_{v,v_i',x,y}\sum_{i=1}^n|u_i(v_i,(v_i',v_{-i}),x,y)-u^*_i(v_i,(v_i',v_{-i}),x,y)|+r< 2r.
	\end{align*}
	
	Eventually, we have
	\begin{equation*}
	    \mathcal{N}_{\infty,1}(\mathcal{R},2r)\le |\mathcal{R}_r|\le |\mathcal{U}_r|= \mathcal{N}_{\infty,1}(\mathcal{U},r).
	\end{equation*}
	The proof is completed.
\end{proof}


We now prove Theorem \ref{gen} and Corollary \ref{corollary}.

\begin{proof}[Proof of Theorem \ref{gen} and Corollary \ref{corollary}] Applying Lemma \ref{lemma} to the spaces $\mathcal{P}$ and $\mathcal{U}$, we have that
	\begin{align*}
		&\mathbb{P}\bigg[\exists \omega \in\Omega:\bigg|\underset{(v,x,y)}{\mathbb{E}}\bigg[\sum\limits_{i=1}^n{p_i^\omega(v,x,y)}\bigg]-\frac{1}{L}\sum\limits_{l=1}^L\sum\limits_{i=1}^n p_{i}^\omega(v^{(l)},x^{(l)},y^{(l)})\bigg|>\epsilon\bigg]\\
		\le& 2\mathcal{N}_{\infty,1} \big(\mathcal{P},\frac{\epsilon}{3}\big)\exp\bigg(-\frac{2L\epsilon^2}{9n^2}\bigg),
	\end{align*}
	and
	\begin{align*}
		&\mathbb{P}\bigg[\exists \omega \in\Omega:\bigg|\underset{(v,x,y)}{\mathbb{E}}\bigg[\sum\limits_{i=1}^n{reg}_i^{\omega}(v,x,y)\bigg]-\frac{1}{L}\sum\limits_{l=1}^L\sum\limits_{i=1}^n {reg}_{i}^\omega(v^{(l)},x^{(l)},y^{(l)})\bigg|>\epsilon\bigg]\\
		\le& 2\mathcal{N}_{\infty,1} \big(\mathcal{R},\frac{\epsilon}{3}\big)\exp\bigg(-\frac{2L\epsilon^2}{9n^2}\bigg)\\
		\le&2\mathcal{N}_{\infty,1} \big(\mathcal{U},\frac{\epsilon}{6}\big)\exp\bigg(-\frac{2L\epsilon^2}{9n^2}\bigg),
	\end{align*}
	where the last inequality follows from Lemma \ref{lemma2}.
	
	Further, we assume that
	\begin{equation*}
	\epsilon\ge\sqrt{\frac{9n^2}{2L}\bigg(\log\frac{4}{\delta}+\max\Big\{\log\mathcal{N}_{\infty,1}\Big(\mathcal{P},\frac{\epsilon}{3}\Big),log\mathcal{N}_{\infty,1}\Big(\mathcal{U},\frac{\epsilon}{6}\Big)\Big\}\bigg)}.
	\end{equation*}
	Then, we have the following inequalities,
	\begin{gather*}
		\mathbb{P}\bigg[\exists \omega \in\Omega:\bigg|\underset{(v,x,y)}{\mathbb{E}}\bigg[\sum\limits_{i=1}^n{p_i^\omega(v,x,y)}\bigg]-\frac{1}{L}\sum\limits_{l=1}^L\sum\limits_{i=1}^n p_{i}^\omega(v^{(l)},x^{(l)},y^{(l)})\bigg|>\epsilon\bigg]\le\frac{\delta}{2},
	\end{gather*}
	and
	\begin{gather*}
		\mathbb{P}\bigg[\exists \omega \in\Omega:\bigg|\underset{(v,x,y)}{\mathbb{E}}\bigg[\sum\limits_{i=1}^n{reg}_i^{\omega}(v,x,y)\bigg]-\frac{1}{L}\sum\limits_{l=1}^L\sum\limits_{i=1}^n {reg}_{i}^\omega(v^{(l)},x^{(l)},y^{(l)})\bigg|>\epsilon\bigg]\le\frac{\delta}{2}.
	\end{gather*}
	
 Thus, with probability at least $1-\delta$, for any $\omega\in \Omega$, we have that
 \begin{equation}
 	\bigg|\underset{(v,x,y)}{\mathbb{E}}\bigg[\sum\limits_{i=1}^n{p_i^\omega(v,x,y)}\bigg]-\frac{1}{L}\sum\limits_{l=1}^L\sum\limits_{i=1}^n p_{i}^\omega(v^{(l)},x^{(l)},y^{(l)})\bigg|<\epsilon,\label{p1}
 \end{equation}
and
\begin{equation}
 	\bigg|\underset{(v,x,y)}{\mathbb{E}}\bigg[\sum\limits_{i=1}^n{reg}_i^{\omega}(v,x,y)\bigg]-\frac{1}{L}\sum\limits_{l=1}^L\sum\limits_{i=1}^n {reg}_{i}^\omega(v^{(l)},x^{(l)},y^{(l)})\bigg|<\epsilon.\label{p2}
\end{equation}
 Equivalently, when the number of samples $L$ is large enough, {\it i.e.},
 \begin{equation*}
     L\ge \frac{9n^2}{2\epsilon^2}\Big(\log\frac{4}{\delta}+\max\Big\{\log\mathcal{N}_{\infty,1}\Big(\mathcal{P},\frac{\epsilon}{3}\Big),log\mathcal{N}_{\infty,1}\Big(\mathcal{U},\frac{\epsilon}{6}\Big)\Big\}\Big),
 \end{equation*}
 then, the eqs. (\ref{p1}) and (\ref{p2}) both hold with probability at least $1-\delta$.
 
 The proof is completed.
\end{proof}

{\subsection{{Proof of the Generalization Bound for Myerson Auctions}}
Denote $rev(v,x,y)$ as $\sum_{i=1}^np_i(v,x,y)$, then we have the following theorem,
\begin{theorem}
Assume the item valuation for each bidder is not larger than $1$. When the sample complexity satisfies
$L\ge\frac{1}{2\epsilon^2}\log\frac{2}{\delta}$,
with probability at least $1-\delta$, we have 
\begin{equation*}
    \bigg|\frac{1}{L}\sum_{\ell=1}^Lrev(v^{(\ell)},x^{(\ell)},y^{(\ell)})-\mathbb{E}_{(v,x,y)}\Big[rev(v,x,y)\Big]\bigg|\le\epsilon.
\end{equation*}
\end{theorem}
\begin{proof}
    Since $v_i\le 1$, we have 
    \begin{equation*}
        rev(v,x,y)=\sum_{i=1}^n p_i(v,x,y)\le\sum_{i=1}^n b_i(v,x,y)v_i\le\sum_{i=1}^n b_i(v,x,y)\le 1.
    \end{equation*}
    According to the Hoeffding’s inequality, we have

\begin{equation*}
\mathbb{P}\bigg[\bigg|\frac{1}{L}\sum_{\ell=1}^L rev(v^{(\ell)},x^{(\ell)},y^{(\ell)})-\mathbb{E}_{(v,x,y)}\Big[rev(v,x,y)\Big]\bigg|\ge \epsilon\bigg]\le 2\exp(-2L\epsilon^2).
\end{equation*}
Let $2\exp(-2L\epsilon^2)\le\delta$, then we obtain what we need. The proof is completed.
\end{proof}

\subsection{{Orbit Averaging over Subsets of Bidders/Items}}
In addition, we can extended our theory to orbit averaging over the subset of the bidders/items.
\begin{theorem}
Let $\mathcal{Q}$ be the orbit averaging over any subset of bidders and items, and $(g,p)$ be any mechanism. Then we have
\begin{equation*}
\mathbb{E}_{(v,x,y)}\bigg[\sum\limits_{i=1}^n [\mathcal{Q}{p}]_i(v,x,y)\bigg]=\mathbb{E}_{(v,x,y)}\bigg[\sum\limits_{i=1}^n{p}_i(v,x,y)\bigg],\end{equation*}
and
\begin{equation*}
\mathbb{E}_{(v,x,y)}\bigg[\sum\limits_{i=1}^n{reg}_i(v,x,y)\bigg]\ge\mathbb{E}_{(v,x,y)}\bigg[\sum\limits_{i=1}^n[\mathcal{Q}{reg}]_i(v,x,y)\bigg],
\end{equation*}
where $reg_i$ is the ex-post regret of bidder $i$.
\end{theorem}
\begin{proof}
    Without loss of generality, we assume that $\mathcal{Q}_1$ takes average over the first $\tilde n$ bidders, $\mathcal{Q}_2$ takes average over the first $\tilde m$ items and $\mathcal{Q}_3$ takes average over the first $\tilde n$ bidders and $\tilde m$ items. Denote $Z_1$ as $(v_{ij},x_i:i>\tilde n,j\in[m])$ and $Z_2$ as $(v_{i j},y_j:i\in[n],j>\tilde{m})$. Following Theorem \ref{the}, we have
    \begin{equation*}
        \mathbb{E}\bigg[\sum\limits_{i=1}^n [\mathcal{Q}_1{p}]_i(v,x,y)\bigg|Z_1\bigg]=\mathbb{E}\bigg[\sum\limits_{i=1}^n{p}_i(v,x,y)\bigg|Z_1\bigg],
    \end{equation*}
    and
    \begin{equation*}
        \mathbb{E}\bigg[\sum\limits_{i=1}^n{reg}_i(v,x,y)\bigg|Z_1\bigg]\ge\mathbb{E}\bigg[\sum\limits_{i=1}^n[\mathcal{Q}_1{reg}]_i(v,x,y)\bigg|Z_1\bigg].
    \end{equation*}
    Then, combining the fact that $\mathbb{E}[\mathbb{E}[X|Y]]=\mathbb{E}[X]$, we have
        \begin{equation*}
        \mathbb{E}\bigg[\sum\limits_{i=1}^n [\mathcal{Q}_1{p}]_i(v,x,y)\bigg]=\mathbb{E}\bigg[\sum\limits_{i=1}^n{p}_i(v,x,y)\bigg],
    \end{equation*}
    and
    \begin{equation*}
        \mathbb{E}\bigg[\sum\limits_{i=1}^n{reg}_i(v,x,y)\bigg]\ge\mathbb{E}\bigg[\sum\limits_{i=1}^n[\mathcal{Q}_1{reg}]_i(v,x,y)\bigg].
    \end{equation*}
    Similarly, replace $Z_1$ by $Z_2$, we can obtain the equations all hold for $\mathcal{Q}_2$.
    
    Finally, we prove that $\mathcal{Q}_3=\mathcal{Q}_1\mathcal{Q}_2=\mathcal{Q}_2\mathcal{Q}_1$. The proof is same with the proof of Lemma \ref{lemma:composition}. Only replace $n$ and $m$ by $\tilde n$ and $\tilde m$ respectively, and we obtain the result. 
    
    The proof is completed.
\end{proof}
\begin{theorem}
Let $\mathcal{Q}$ be the orbit averaging over any subset of bidders and items, and $\mathcal{U}=\{u^\omega:\omega\in\Omega\}$ and $\mathcal{P}=\{p^\omega:\omega\in\Omega\}$ the sets of all possible utilities and payment rules. Then we have
\begin{equation*}
\mathcal{N}_{\infty,1}(\mathcal{Q}{\mathcal{U}},r)\le\mathcal{N}_{\infty,1}(\mathcal{U},r)~~\text{and}~~\mathcal{N}_{\infty,1}(\mathcal{Q}{\mathcal{P}},r)\le\mathcal{N}_{\infty,1}(\mathcal{P},r),
\end{equation*}
where $\mathcal{N}_{\infty,1}({\mathcal{U}},r)$ and $\mathcal{N}_{\infty,1}({\mathcal{P}},r)$ are the minimum numbers of balls with radius $r$ that can cover $\mathcal{U}$ and $\mathcal{P}$ under $l_{\infty,1}$-distance, respectively.
\end{theorem}
\begin{proof}
 Without loss of generality, we assume that $\mathcal{Q}_1$ takes average over the first $\tilde n$ bidders, $\mathcal{Q}_2$ takes average over the first $\tilde m$ items and $\mathcal{Q}_3$ takes average over the first $\tilde n$ bidders and $\tilde m$ items.
 
    We can prove the distance between two mechanisms becomes smaller after orbit averaging, {\it i.e.},
    \begin{equation*}
        l_{\infty,1}(\mathcal{Q}p,\mathcal{Q}p')\le l_{\infty,1}(p,p')~~\text{and}~~l_{\infty,1}(\mathcal{Q}u,\mathcal{Q}u')\le l_{\infty,1}(u,u').
    \end{equation*}
    Only replace $n$ and $m$ by $\tilde n$ and $\tilde m$ in the proof of Theorem \ref{covering}, and we obtain the results.
    
    The proof is completed.
\end{proof}
}

{

\subsection{Average over Subgroups}
It is worth noting that our proof only replies one assumption that the valuation joint distribution is invariant under the bidder/item permutation. Consequently, we may adopt orbit averaging over any subgroup of $S_n\times S_m$, while {the benefits on revenue and ex-post regret still hold}. Hence, there is a trade-off between the auction mechanism performance (revenue and ex-post regret) and the computational complexity: better performance requires more computation.
In addition, the choice of the subgroup can also depend on the input feature $x$ \cite{puny2022frame}, which could be more flexible.}
\section{Additional Experimental Details}\label{exdetails}

This section presents additional experimental details and results omitted from the main text due to space limitation. 

\subsection{Additional Experimental Settings}

In this section, we present detailed experimental settings. 

\subsubsection{Network Architectures}

We first describe the RegretNet's architecture \cite{dutting2019optimal}. A RegretNet consists of two parts: the allocation network $g^\omega:\mathbb{R}^{nm}\to[0,1]^{nm}$ and the payment network $p^\omega:\mathbb{R}^{nm}\to \mathbb{R}_{\ge0}^n$, both of which are modeled as three-layer fully-connected networks with {\it tanh} activations. Every layer in the two networks includes $100$ nodes. 

For each item $j$, the payment network outputs a probability vector $(g_{1j}^\omega(b),\dots,g_{nj}^\omega(b))^T$, where $g_{ij}^\omega(b)$ is the probability of allocating the item $j$ to the bidder $i$. To avoid allocating one item over once, a feasible allocation network needs to satisfy $\sum_{i=1}^ng_{ij}^\omega(b)\le 1$ for all $j\in[m]$, $\omega\in\Omega$, and $b\in\mathcal{V}^{nm}$. Therefore, we compute the allocation via a {\it softmax activation function}. In addition, to present the probability that the item is reserved, an extra dummy node is included in the softmax computation.

To ensure the {\it individual rational} condition, the payment network $p^\omega$ is required to output a payment vector $p^\omega(b)$, such that $p_i^\omega(b)\le \sum_{j=1}^mg_{ij}^\omega(b)b_{i j}$ for all $i\in[n]$. Therefore, the payment network first computes a fractional
payment $\overline{p}_i^\omega(b)\in[0,1]$ for each bidder $i$ using a sigmoidal unit. Then, the final payment of the bidder $i$ is 
\begin{equation*}
    p^\omega_i(b)=\overline{p}_i^\omega(b)\sum_{j=1}^mg_{i j}^\omega(b)b_{i j}\le\sum_{j=1}^mg_{ij}^\omega(b)b_{i j}.
\end{equation*}
An overview of the RegretNet's architecture is illustrated in the following Figure \ref{fig:architecture}.

\begin{figure}[h]
    \centering
    \includegraphics[scale=0.47]{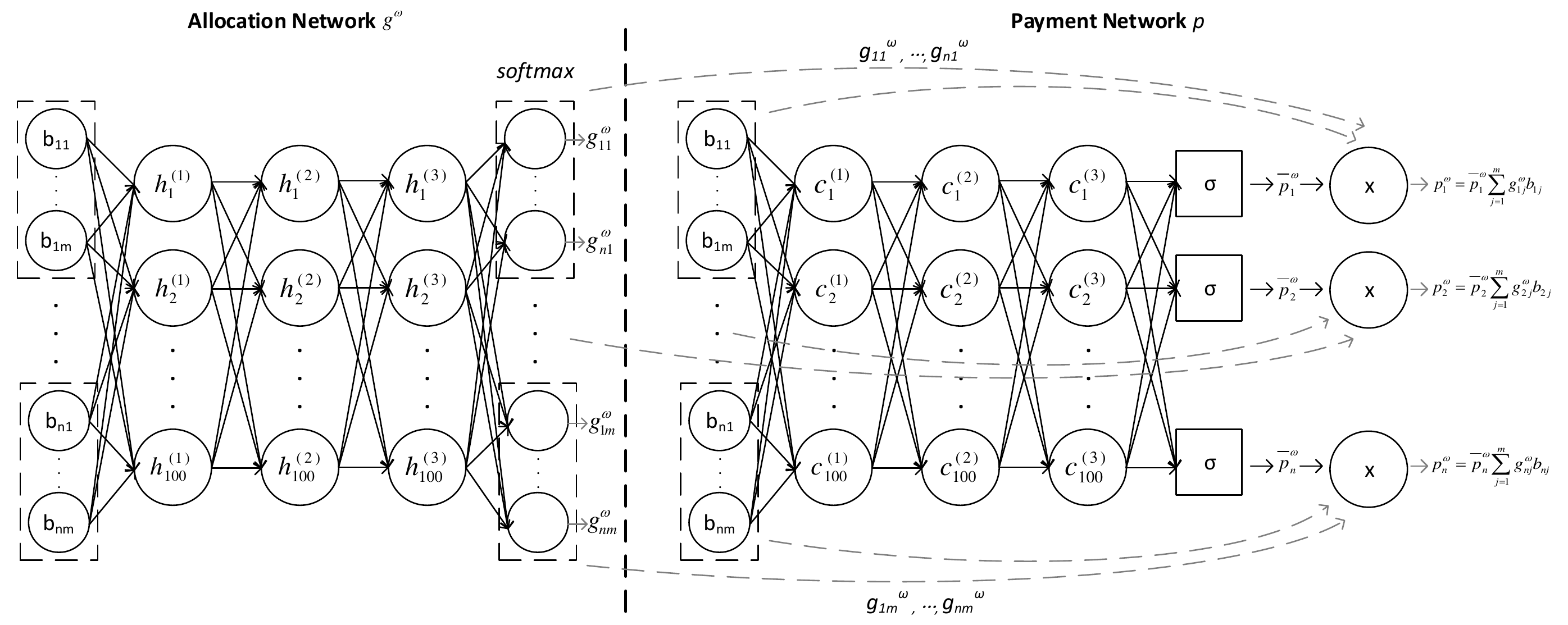}
    \caption{The allocation network $g^\omega$ and the payment network $p^\omega$ of the RegretNet with the overall mechanism parameter $\omega=(\omega_g,\omega_p)$.}
    \label{fig:architecture}
\end{figure}

RegretNet-PE is designed by modifying RegretNet. We adopt the allocation rule as $\widetilde{g}^\omega=\mathcal{Q}_3g^\omega$ and the payment rule as $\widetilde{p}^\omega=\mathcal{Q}_3p^\omega$, respectively. In this way, we may guarantee that in RegretNet-PE, the allocation is feasible and the mechanism is {\it individual rational}, {\it i.e.}, $\sum_{i=1}^n\widetilde{g}_{ij}^\omega(b)\le 1$, and $\widetilde{p}^\omega_i(b)\le\sum_{j=1}^m\widetilde{g}^\omega_{ij}(b)b_{ij}$. We may also show that the RegretNet-PE is always permutation-equivariant and has the same number of coefficients as the RegretNet. The proof can be found in 
Appendix \ref{A.1}.

\subsubsection{Training Procedures}

We adopt the augmented Lagrangian method to minimize the following object function with a
quadratic penalty term for violating the constraints,
\begin{equation*}
	\mathcal{L}_{\rho}(\omega,\lambda)=-\frac{1}{L}\sum_{l=1}^L\sum_{i=1}^np^\omega_i(v^{(l)})+\sum_{i=1}^n \lambda_i \widehat{reg}_i(\omega)+\frac{\rho}{2}\Big(\sum_{i=1}^n\widehat{reg}_i\Big)^2,
\end{equation*}
where $L$ is the number of samples, $\lambda$ is a vector of Lagrange multipliers, and $\rho>0$ is a parameter to control the weight of the quadratic penalty. We alternately update the overall mechanism parameter $\omega$ and the Lagrange multiplier $\lambda$ as follows:

(a) $\omega^{new}\in\arg\min_{\omega} \mathcal{L}_{\rho}(\omega^{old},\lambda^{old})$  (every iteration);

(b) $\lambda^{new}_i=\lambda_i^{old}+\rho\cdot\widehat{reg}_i(\omega^{new})$, $\forall i\in[n]$ (every $T_\lambda$ iterations).

The training procedure is described in the following Algorithm \ref{algorithm}. 

\begin{algorithm}[h]\label{algorithm}
	\caption{RegretNet and RegretNet-PE Training}
	\KwIn{Batches $\mathcal{S}_1,\dots,\mathcal{S}_T$ of size $B$}
	\textbf{Parameters: }$\forall t\in [T]$, $\rho_t>0$, $\gamma>0$, $\eta>0$, $T\in\mathbb{N}$, $R\in\mathbb{N}$, $T_\lambda\in\mathbb{N}$\\
	\textbf{Initialize: }$\omega^0\in\mathbb{R}^d$, $\lambda^0\in\mathbb{R}^n$\\
	\For{$t = 0$ \KwTo $T$}
	{Receive batch $S_t=\{v^{(1)},\dots,v^{(B)}\}$\\
		Initialize misreports \ ${v'}_i^{(l)} \in \mathcal{V}^m$, $\forall l\in [B]$, $i\in[n]$\\
		\For{$r = 0$ \KwTo $R$}{
			$\forall l\in[B]$, $i\in [n]$:\\
			\ \ \ \ \ ${v'}_i^{(l)} \leftarrow {v'}_i^{(l)}+\gamma\nabla_{v_i'}u_i^\omega\big(v_i^{(l)},({v'}_i^{(l)},v_{-i}^{(l)})\big)$
		}
		Compute  ex-post regret  gradient :
		{$\forall l\in[B]$, $i\in [n]$:\\
			\ \ \ \ \ $g_{l,i}^t \leftarrow \nabla_{\omega}\Big[u_i^\omega\big(v_i^{(l)},({v'}_i^{(l)},v_{-i}^{(l)})\big)-u_i^\omega\big(v_i^{(l)},v^{(l)}\big)\Big]\Big|_{\omega=\omega^t}$
		}\\
		Compute Lagrangian gradient using Equation \ref{lag} and update $\omega^t$:\\
		\ \ \ \ \ {$\omega^{t+1}\leftarrow\omega^t-\eta \nabla_{\omega}\mathcal{L}_{\rho_t}(\omega^t,\lambda^t)$}\\
		Update Lagrange multipliers $\lambda$ once in $T_\lambda$ iterations:\\
		\eIf{$t$ is a multiple of $T_\lambda$}
		{$\lambda^{t+1}_i=\lambda^t_i+\rho_t\widehat{reg}_i(\omega^{t+1})$, $\forall i\in[n]$}
		{$\lambda^{t+1}=\lambda^t$}
	}
\end{algorithm}

We divide all training samples $\mathcal{S} $ into $T$ batches $\mathcal{S}_1,\dots,\mathcal{S}_T$ of size $B$. At iteration $t$, we use the batch $\mathcal{S}_t=\{v^{(1)},\dots,v^{(B)}\}$. 

The update (a) is computed via Adam. The gradient of $\mathcal{L}_\rho$ {\it w.r.t.} $\omega$ for a fixed $\lambda^t$ is as below,
\begin{equation*}
	\nabla_\omega\mathcal{L}_\rho(\omega,\lambda^t)=-\frac{1}{B}\sum_{l=1}^B\sum_{i=1}^n\nabla_{\omega}p_i^{\omega}(v^{(l)})+\sum_{l=1}^B\sum_{i=1}^n\lambda^t_ig_{l,i}^t+\rho\sum_{l=1}^B\sum_{i=1}^n\widehat{reg}_i(\omega)g_{l,i}^t,
\end{equation*}
where
\begin{gather*}
	\widehat{reg}_i(\omega)=\frac{1}{B}\sum_{l=1}^B\max_{v_i'\in\mathcal{V}^m}u_i^\omega(v_i^{(l)},(v_i',v_{-i}^{(l)}))-u_i^\omega(v_i^{(l)},v^{(l)}),
\end{gather*} 
and
\begin{gather*}
	g_{l,i}^t = \nabla_{\omega} \Big[\max_{v_i'\in\mathcal{V}^m}u_i^\omega(v_i^{(l)},(v_i',v_{-i}^{(l)}))-u_i^\omega(v_i^{(l)},v^{(l)})\Big]\Big|_{\omega=\omega^t}.
\end{gather*} 
	 Because $\widehat{reg}_i(\omega) $ and $g_{l,i}^t$ both contain a ``max'' over misreports\footnote{The misreport refers to an arbitrary bid, rather than restricted to be a truthful bid \cite{dutting2019optimal}.}, we use another Adam to compute the approximated best biddings ${v'}^{(l)}$. In each update on $\omega^t$, we perform $R$ updates to compute a best bidding ${v'}^{(l)}_i$ for each $i\in[n]$. In particular, we maintain the misreports ${v'}^{(l)}$ for each sample $l$ as the initial value in the next iteration. Then, we use these biddings ${v'}^{(l)}$ to compute the gradient $\nabla_\omega\mathcal{L}_\rho(\omega,\lambda^t)$ and then, update $\omega^t$ as $\omega^{t+1}=\omega^t-\eta\nabla_\omega\mathcal{L}_\rho(\omega^t,\lambda^t)$. 
	 After every $T_\lambda$ iterations, we update $\lambda^t$ as $\lambda^{t+1}_i=\lambda^t_i+\rho\widehat{reg}_i(\omega^{t+1})$. 
	 In addition, we increase the value of $\rho$ every a certain number of iterations, where we set the value of $\rho_t$ in each iteration $t$ prior to training.
\begin{table}[t]
	\centering{
		\caption{Additional experimental results. "$n\times m$ Normal" refers that there are $n$ bidders and $m$ items, and the valuation is drawn from the truncated normal distribution $\mathcal{N}(0.3,0.1)$ in [0,1]. The true values of the ex-post regret and the generalization error (GE) are the products of the values in the table and a factor of $10^{-5}$.}
		\label{table2}
		\setlength{\tabcolsep}{1.6mm}
		\renewcommand{\arraystretch}{1.2}
		
		\begin{tabular}{c|ccc|ccc}
			\toprule
			\multirow{2}{*}{Method} & \multicolumn{3}{c|}{$2\times1$ Normal}            & \multicolumn{3}{c}{$3\times1$ Normal}.  \\
			& Revenue & Regret & GE & Revenue & Regret & GE  \\
			\midrule
		Optimal                 & $0.304$   & $0$    & -  & $0.391$  & 0   & - \\
			\midrule
			RegretNet               & $0.275$   & $97.0$ & $8.50$ & $0.321$ & $84.0$ & $45.5$\\
			RegretNet-Test          & $0.275$   & $95.2$ & -    & $0.321$ & $75.0$ & - \\
			RegretNet-PE            & $0.276$   & $85.4$ & $8.40$  & $0.382$ &
			$69.7$ & $27.6$ \\
			\midrule
			\multirow{2}{*}{Method} & \multicolumn{3}{c|}{$2\times2$ Normal} & \multicolumn{3}{c}{$5\times3$ Normal}\\
			& Revenue  & Regret & GE  & Revenue & Regret  & GE\\ 
			\midrule
			RegretNet    &  $0.577$ & $343$ & $246$ & $1.05$  & $114$ & $77.0$ \\
			RegretNet-Test  &  $0.577$ & $327$ & - & $1.05$  & $32.0$ & - \\
			RegretNet-PE    &  $0.577$ & $318$ & $77.0$ & $1.09$  & $75.0$ & $70.0$ \\
			\bottomrule
			
		\end{tabular}
	}
\end{table}
\begin{table}[t]
	\centering{
		\caption{Additional experimental result, where "$n\times m$ Compound" refers that there are $n$ bidders and $m$ items, and the valuations are i.i.d. sampled from the compound distributions.}
		\label{table3}
		\setlength{\tabcolsep}{1.6mm}
		\renewcommand{\arraystretch}{1.2}
		
		\begin{tabular}{c|cc|cc}
			\toprule
			\multirow{2}{*}{Method} & \multicolumn{2}{c|}{$3\times1$ Compound}     &   \multicolumn{2}{c}{$5\times1$ Compound}    \\
			& Revenue & Regret & Revenue & Regret \\          
			\midrule
			RegretNet              	  & $0.516$   & $<0.001$ & $0.329$ & $<0.001$  \\
			EquivariantNet          & $0.498$   & $<0.001$ &  $0.311$    & $<0.001$  \\
			RegretNet-PE            & $0.539$   & $<0.001$ & $0.356$  & $<0.001$ \\
			\bottomrule
		\end{tabular}
	}
\end{table}

\begin{table}[t]
	\centering{
		\caption{Additional experimental result, where "$n\times m$ Uniform" refers that there are $n$ bidders and $m$ items, and the valuations are i.i.d. sampled from the uniform distribution $U[0,1]$.}
		\label{table5}
		\setlength{\tabcolsep}{1.6mm}
		\renewcommand{\arraystretch}{1.2}
		
		\begin{tabular}{c|ccc|ccc}
			\toprule
			\multirow{2}{*}{Method} &  \multicolumn{3}{c|}{$2\times5$ Uniform} & \multicolumn{3}{c}{$5\times3$ Uniform} \\
       & Revenue & Regret  & GE & Revenue & Regret  & GE \\
\midrule
RegretNet   & $2.24$  & $104$ & $86.4$  & ${1.56}$ & $28.4$ & ${19.4}$\\
RegretNet-Test & $2.24$  & $74.0$ & - & ${1.56}$ & $8.60$ & -\\
RegretNet-PE   & $2.38$  & $89.9$ & $24.9$ & $1.85$ & ${20.1}$ & ${11.8}$\\
			\bottomrule
		\end{tabular}
	}
\end{table}



\begin{figure}
    \centering
    \includegraphics[scale=0.77]{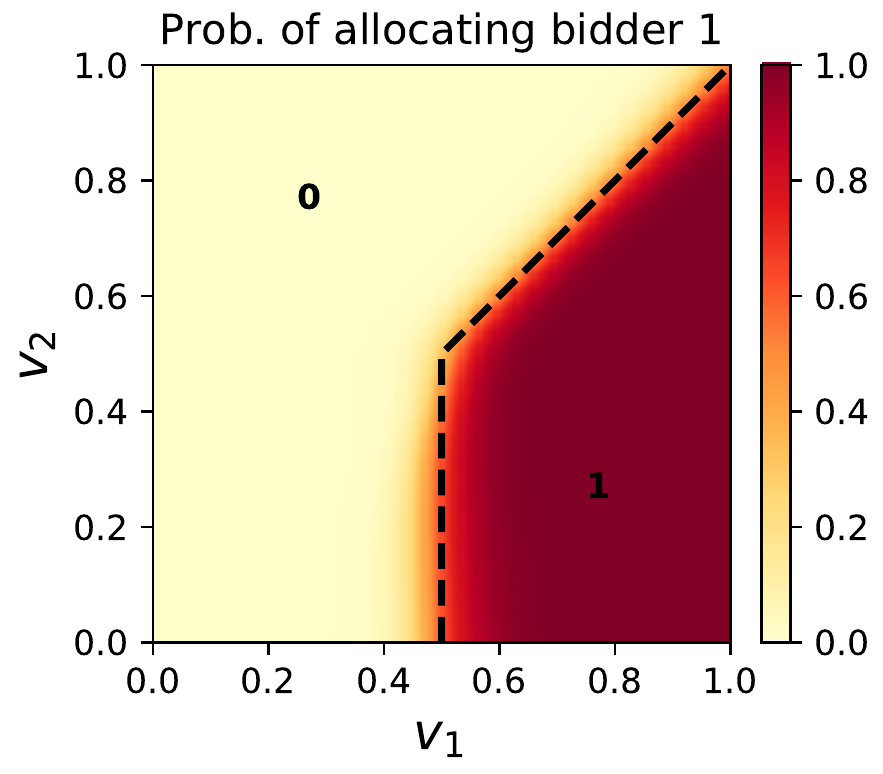}
    \includegraphics[scale=0.77]{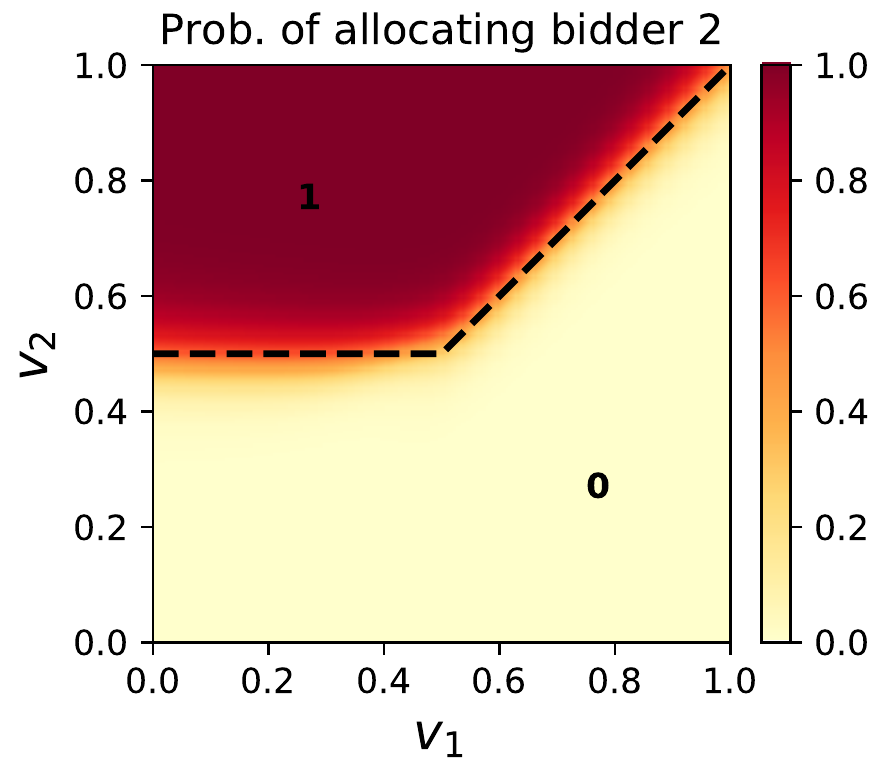}\\
    \includegraphics[scale=0.77]{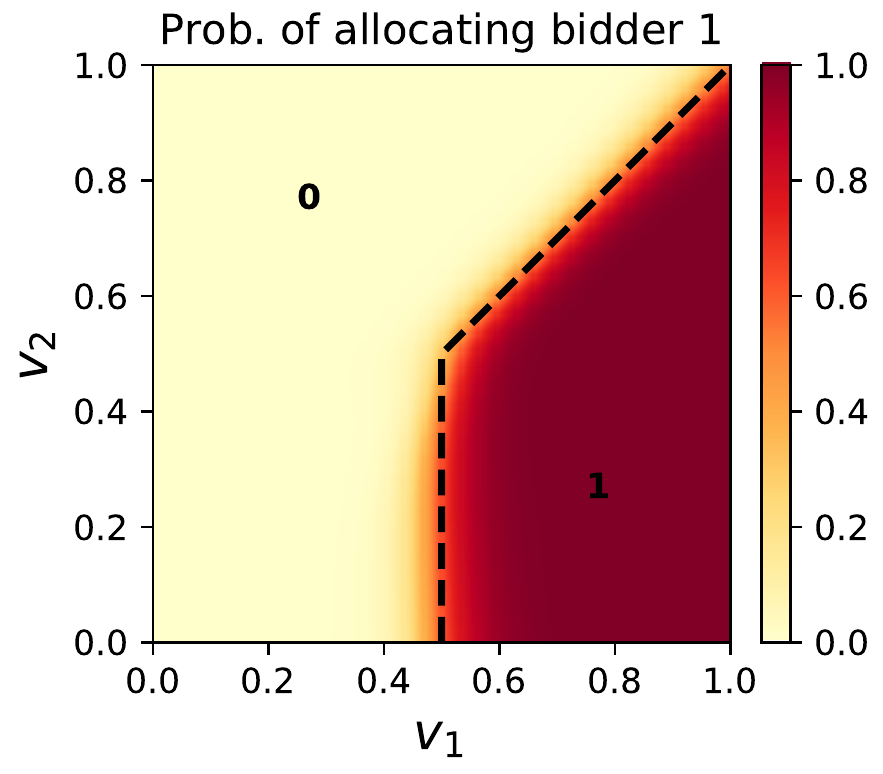}
    \includegraphics[scale=0.77]{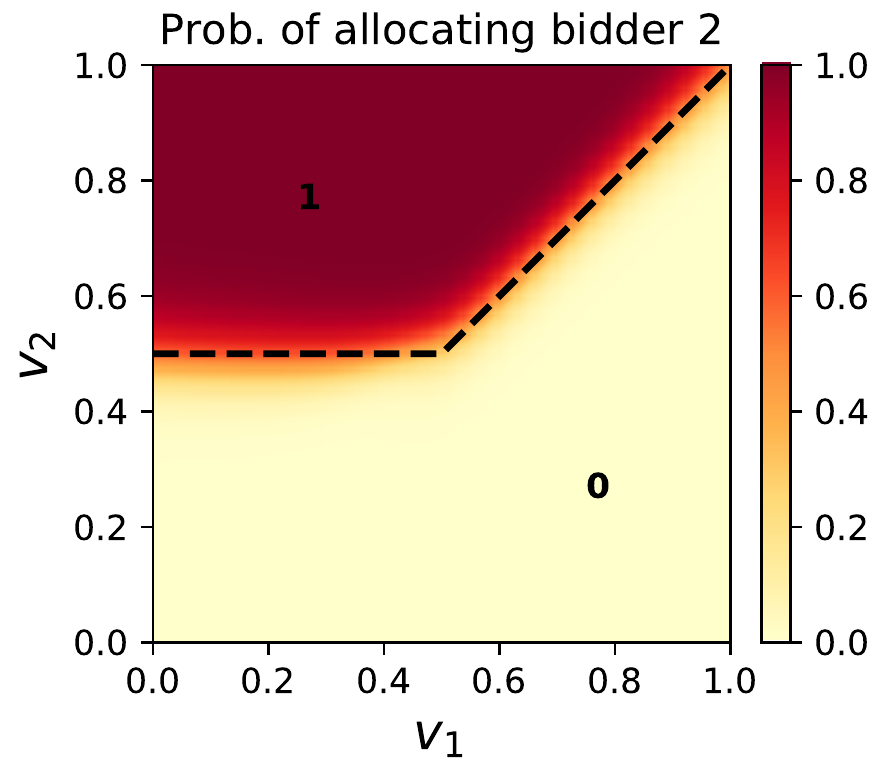}
    \caption{Allocation rule learned by RegretNet (up) and RegretNet-PE (down) for two-bidder and one-item setting. The solid regions describe the probability to allocating the item to bidder $1$ (left) and bidder $2$ (right). The optimal auction mechanism is described by the regions separated by the dashed black lines, where the number $0$ or $1$ is the probability of optimal allocation rule in the region.}
    \label{figure 2}
\end{figure}

\begin{figure}
    \centering
    \includegraphics[scale=0.77]{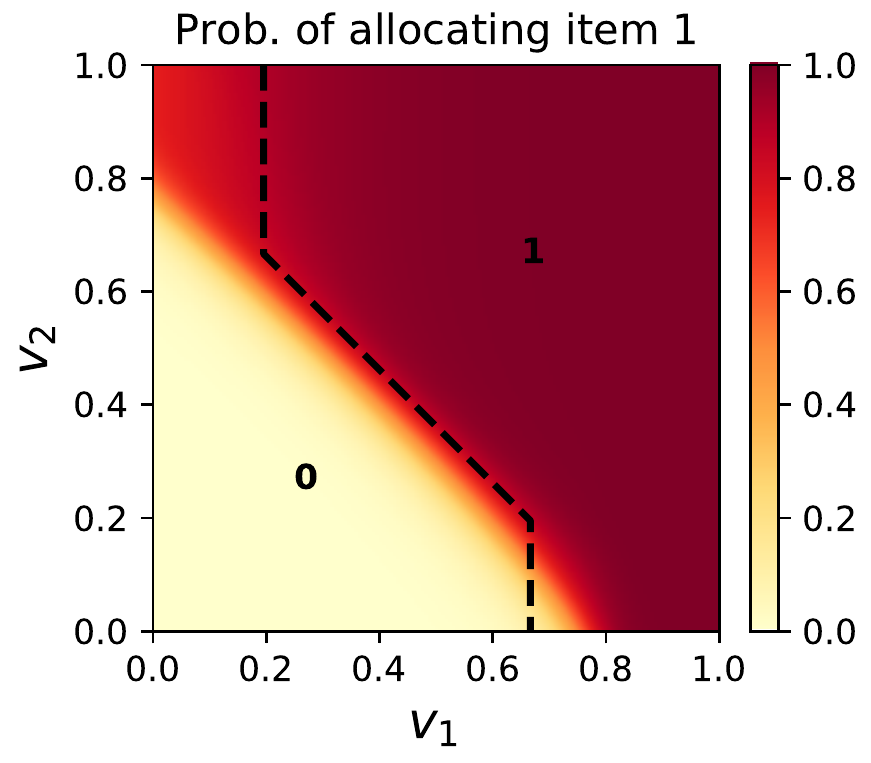}
    \includegraphics[scale=0.77]{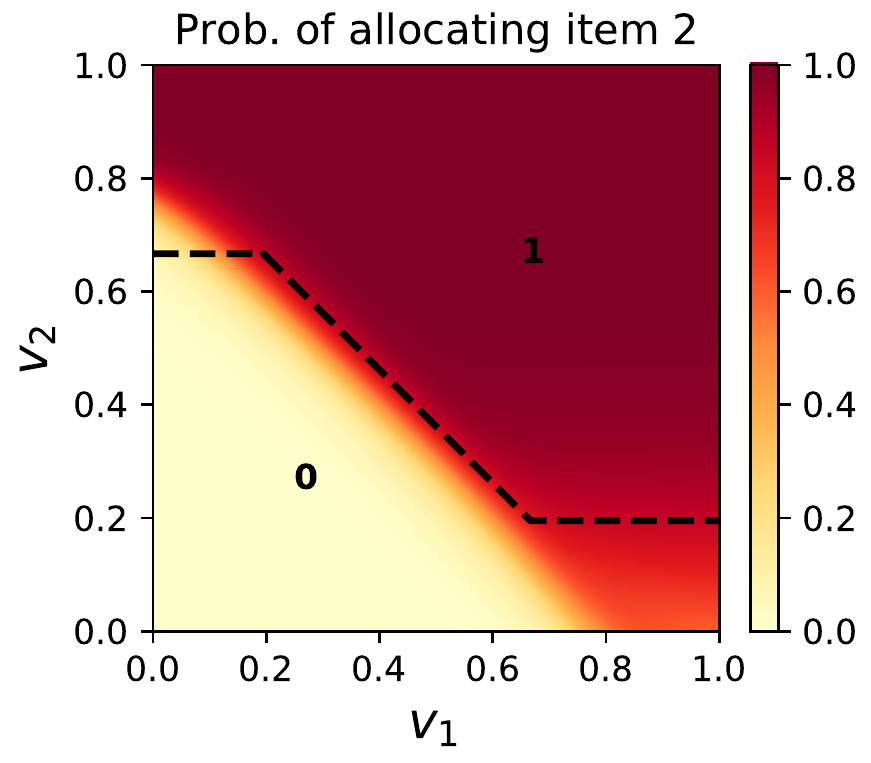}\\
    \includegraphics[scale=0.77]{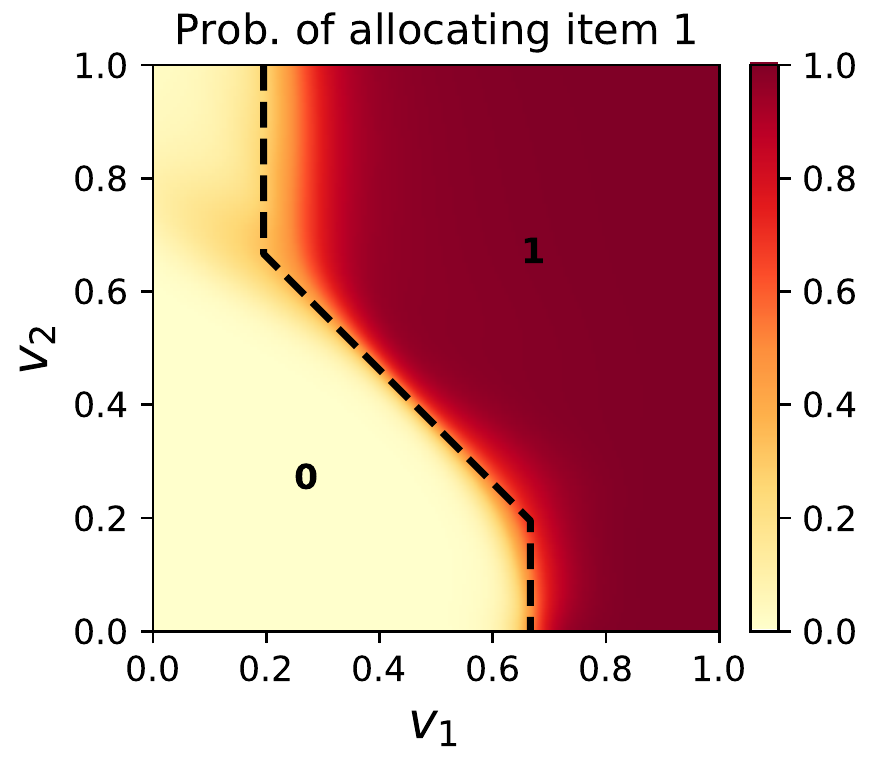}
    \includegraphics[scale=0.77]{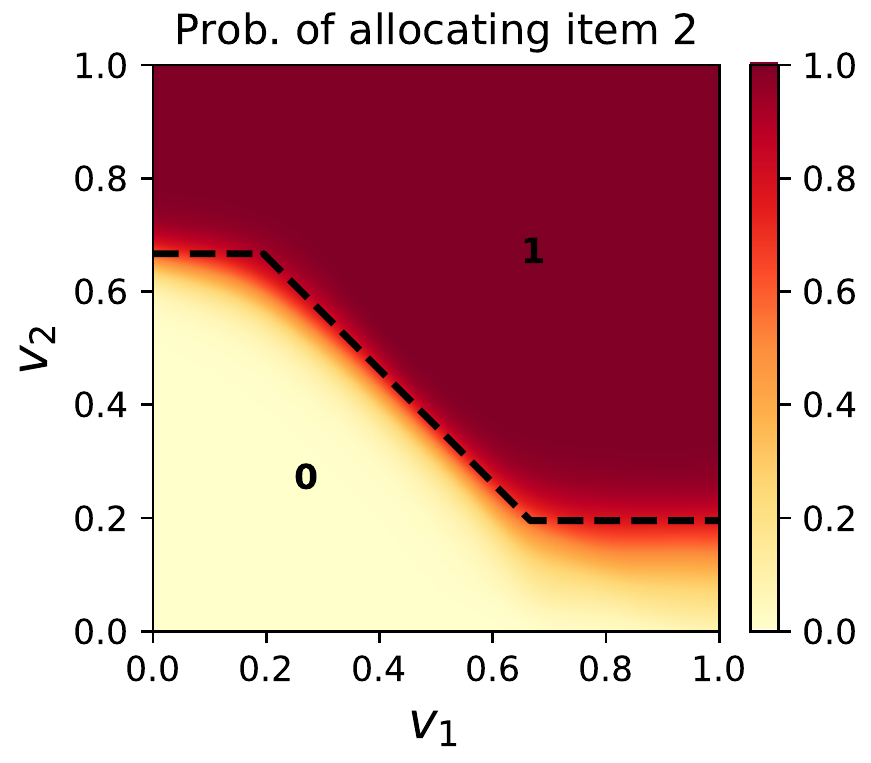}
    \caption{{Allocation rule learned by RegretNet (up) and RegretNet-PE (down) for one-bidder and two-item setting. The solid regions describe the probability of allocating the first item (left) and the second item (right). The optimal auction mechanism is described by the regions separated by the dashed black lines, where the number $0$ or $1$ is the probability of optimal allocation rule in the region.}}
    \label{figure 3}
\end{figure}

\subsubsection{Test Settings}
To verify our Theorem \ref{the} and Theorem \ref{theorem plus}, we first train a RegretNet and then project the will-trained RegretNet to be permutation-equivariant through bidder-item aggregated averaging $\mathcal{Q}_3$, denoted as ``RegretNet-Test''. To meet the {\it symmetric valuation} condition in Theorem \ref{the} and Theorem \ref{theorem plus}, we sample a set of valuations from the distribution, which is denoted by $\mathcal{S}$, and then, induce a set of symmetric samples $\widetilde{\mathcal{S}}=\{\sigma_nv\sigma_m:\sigma_n\in S_n,v\in \mathcal{S},\sigma_m\in S_m\}$ for test.

For a RegretNet-PE, there is no difference between test on $\tilde{\mathcal{S}}$ and on $\mathcal{S}$, because
\begin{equation*}
\frac{1}{n!m!L}\sum_{l=1}^{L}\sum_{\sigma_n\in S_n}\sum_{\sigma_m\in S_m} f(\sigma_n v_i\sigma_m)=\frac{1}{n!m!L}\sum_{l=1}^{L}\sum_{\sigma_n\in S_n}\sum_{\sigma_m\in S_m} f(v_i)=\frac{1}{L}\sum_{l=1}^L f(v_i),
\end{equation*}
for any permutation-equivariant function $f$ and a RegretNet-PE is always permutation-equivariant. 

 To compute the best bidding $v_i'$ for each bidder $i$, we first randomly initialize $1,000$ misreports in all settings, and then, perform {$2,000$} updates on each misreport via Adam with the same settings. Finally, we choose the best one (which induces a maximal utility of the bidder $i$) as the approximated best bidding $v_i'$.
 
\subsubsection{Implementation Details}

We train the models (RegretNet and RegretNet-PE) for up to {150} epochs with a batch size of 128 ($B=128$) and report the early-stop results for RegretNet-PE to obtain a comparable ex-post regret. {The terminal iteration numbers for RegretNet-PE are $10,000$ in the $2\times 1$ setting, $17,000$ in the $3\times 1$ setting, $18,000$ in the $5\times 1$ setting, $300,000$ in the $1\times 2$ setting, and $600,000$ in the $2\times 2$ setting. Our insight is that the larger terminal iteration number required in the $2\times 2$ setting is because of the small model size, {\it i.e.}, where the networks have three layers (each of $100$ nodes). } The value of $\rho$ is initialized as $1.0$ and increased by $5$ every $200$ batches. For each update on $\omega^t$, we initialize one misreport and update the misreport by Adam for each bidder with $25$ steps ($R=25$) and learning rate $0.1$ ($\gamma=0.1$). The final optimal misreports will be used to initialize the misreports for the same batch in the next epoch. We update $\omega^t$ via Adam for every batch with a learning rate of {0.001}. Besides, we update $\lambda^t$ every {200 batches}.

\subsection{Additional Experiment Results}\label{B.2}

We present additional experimental results.
Each valuation $v_{ij}$ is sampled independently from (1) a truncated normal distribution {$\mathcal{N}(0.3,0.1)$} in $[0,1]$; (2) a compound distribution $\mathcal{N}(\frac{x_i}{6},0.1)$ truncated in $[0,1]$, where $x_i$ is sampled independently and uniformly from $\{1,2,3,4,5\}$ (cf. Setting A, \cite{duan2022context}); and (3) a compound distribution $U[0,Sigmoid(x_i^Ty_j)]$, where $x_i$ and $y_j$ are sampled independently and uniformly from $[-1,1]$ (Setting C, \cite{duan2022context}). All results are shown in Table \ref{table2} and Table \ref{table3}. The revenue and ex-post regret of RegretNet and EquivariantNet in Table \ref{table3} come from the previous work \cite{duan2022context}. 
In Table \ref{table3}, we report the ex-post regret as ``$<0.001$'' following the previous works.

{Moreover, we extend our experiments to more complex settings, including two-bidder five-item and five-bidder three-item settings. Due to the computation limitations, we sample $\{3840, 1280\}$ data points and initialize $\{150,120\}$ misreports for test. Each valuation is sampled from the uniform distribution $U[0,1]$ and the truncated normal distribution {$\mathcal{N}(0.3,0.1)$} in $[0,1]$. The results are shown in Tables \ref{table2} and \ref{table5}.} 



{\subsection{{Allocation Rules Learned by RegretNet and RegretNet-PE}}
In this section, we show the allocation rules learned by RegretNet and RegretNet-PE in two-bidder, one-item setting and one-bidder, two-item setting, where the valuation is drawn from the uniform distribution $U[0,1]$. The optimal auction mechanisms are both known.

\subsubsection{Two-bidder and One-item Setting} For the two-bidder, one-item setting, the optimal mechanism is well-known as Myerson auction \cite{myerson1981optimal}, which allocates the item to the highest bidder with receiving a payment of the maximum of the second price and the reserve price, if the highest bid is higher than the reserve price. The allocation rules learned by RegretNet and RegretNet-PE are shown in Figure \ref{figure 2}. From Figure \ref{figure 2}.
We can find that the two learned allocation rules are both almost the same as the optimal mechanism.

\subsubsection{One-bidder and Two-item Setting} The optimal mechanism is given by \cite{manelli2006bundling}. Same with the above, we show the allocation rules learned by RegretNet and that learned by RegretNet-PE in Figure \ref{figure 3}. The improvement is significant. when one item’s valuation is close to $0$ and another item's is close to $1$, the mechanism learned by RegretNet has a positive probability to allocate the item with the lower valuation to the bidder, while RegretNet-PE and the optimal mechanisms would not.}


\begin{thebibliography}{99}

\bibitem{2011Bayesian}
S.~Alaei.
\newblock Bayesian combinatorial auctions: Expanding single buyer mechanisms to
  many buyers.
\newblock {\em The 16th Annual Symposium on Foundations of Computer Science},
  2011.

\bibitem{2013The}
S.~Alaei, H.~Fu, N.~Haghpanah, and J.~Hartline.
\newblock The simple economics of approximately optimal auctions.
\newblock {\em The 16th Annual Symposium on Foundations of Computer Science,},
  2013.

\bibitem{balcan2016sample}
M.-F.~F. Balcan, T.~Sandholm, and E.~Vitercik.
\newblock Sample complexity of automated mechanism design.
\newblock {\em {Neural Information Processing Systems (NeurIPS)}}, 2016.

\bibitem{2019Auction}
L.~Chen, L.~Wang, and Y.~Lan.
\newblock Auction models with resource pooling in modern supply chain
  management.
\newblock {\em Modern Supply Chain Research and Applications}, 2019.

\bibitem{cole2014sample}
R.~Cole and T.~Roughgarden.
\newblock The sample complexity of revenue maximization.
\newblock In {\em ACM symposium on Theory of computing}, 2014.

\bibitem{conitzer2002complexity}
V.~Conitzer and T.~Sandholm.
\newblock Complexity of mechanism design.
\newblock {\em Proceedings of the Eighteenth Conference on Uncertainty in
  Artificial Intelligence (UAI)}, 2002.

\bibitem{conitzer2004self}
V.~Conitzer and T.~Sandholm.
\newblock Self-interested automated mechanism design and implications for
  optimal combinatorial auctions.
\newblock In {\em ACM Conference on Electronic Commerce}, 2004.

\bibitem{curry2020certifying}
M.~Curry, P.-Y. Chiang, T.~Goldstein, and J.~Dickerson.
\newblock Certifying strategyproof auction networks.
\newblock {\em {Neural Information Processing Systems (NeurIPS)}}, 2020.

\bibitem{daskalakis2017strong}
C.~Daskalakis, A.~Deckelbaum, and C.~Tzamos.
\newblock Strong duality for a multiple-good monopolist.
\newblock {\em Econometrica}, 2017.

\bibitem{duan2022context}
Z.~Duan, J.~Tang, Y.~Yin, Z.~Feng, X.~Yan, M.~Zaheer, and X.~Deng.
\newblock A context-integrated transformer-based neural network for auction
  design.
\newblock {\em arXiv preprint arXiv:2201.12489}, 2022.

\bibitem{2014Sampling}
S.~Dughmi, L.~Han, and N.~Nisan.
\newblock Sampling and representation complexity of revenue maximization.
\newblock {\em Web and Internet Economic}, 2014.

\bibitem{dutting2019optimal}
P.~D{\"u}tting, Z.~Feng, H.~Narasimhan, D.~Parkes, and S.~S. Ravindranath.
\newblock Optimal auctions through deep learning.
\newblock In {\em {International Conference on Machine Learning (ICML)}}, 2019.

\bibitem{elesedy2021provably}
B.~Elesedy and S.~Zaidi.
\newblock Provably strict generalisation benefit for equivariant models.
\newblock In {\em {International Conference on Machine Learning (ICML)}}, 2021.

\bibitem{feng2018deep}
Z.~Feng, H.~Narasimhan, and D.~C. Parkes.
\newblock Deep learning for revenue-optimal auctions with budgets.
\newblock In {\em the 17th International Conference on Autonomous Agents and
  Multiagent Systems}, 2018.

\bibitem{he2020recent}
F.~He and D.~Tao.
\newblock Recent advances in deep learning theory.
\newblock {\em arXiv preprint arXiv:2012.10931}, 2020.

\bibitem{huang2008auction}
J.~Huang, Z.~Han, M.~Chiang, and H.~V. Poor.
\newblock Auction-based resource allocation for cooperative communications.
\newblock {\em IEEE Journal on Selected Areas in Communications}, 2008.

\bibitem{ivanov2022optimal}
D.~Ivanov, I.~Safiulin, K.~Balabaeva, and I.~Filippov.
\newblock Optimal-er auctions through attention.
\newblock {\em arXiv preprint arXiv:2202.13110}, 2022.

\bibitem{jansen2008sponsored}
B.~J. Jansen and T.~Mullen.
\newblock Sponsored search: an overview of the concept, history, and
  technology.
\newblock {\em International Journal of Electronic Business}, 2008.

\bibitem{lyle2020benefits}
C.~Lyle, M.~van~der Wilk, M.~Kwiatkowska, Y.~Gal, and B.~Bloem-Reddy.
\newblock On the benefits of invariance in neural networks.
\newblock {\em arXiv preprint arXiv:2005.00178}, 2020.

\bibitem{manelli2006bundling}
A.~M. Manelli and D.~R. Vincent.
\newblock Bundling as an optimal selling mechanism for a multiple-good
  monopolist.
\newblock {\em Journal of Economic Theory}, 2006.

\bibitem{2013Learning}
M.~Mohri and A.~M. Medina.
\newblock Learning theory and algorithms for revenue optimization in
  second-price auctions with reserve.
\newblock {\em {International Conference on Machine Learning (ICML)}}, 2013.

\bibitem{mohri2018foundations}
M.~Mohri, A.~Rostamizadeh, and A.~Talwalkar.
\newblock {\em Foundations of machine learning}.
\newblock MIT press, 2018.

\bibitem{myerson1981optimal}
R.~B. Myerson.
\newblock Optimal auction design.
\newblock {\em Mathematics of operations research}, 1981.

\bibitem{2016A}
H.~Narasimhan and D.~C. Parkes.
\newblock A general statistical framework for designing strategy-proof
  assignment mechanisms.
\newblock {\em {Uncertainty in Artificial Intelligence (UAI)}}, 2016.

\bibitem{2007Algorithmic}
N.~Nisan, T.~Roughgarden, E.~Tardos, and V.~V. Vazirani.
\newblock Algorithmic game theory.
\newblock Cambridge University Press, 2007.

\bibitem{pavlov2011optimal}
G.~Pavlov.
\newblock Optimal mechanism for selling two goods.
\newblock {\em The BE Journal of Theoretical Economics}, 2011.

\bibitem{2012An}
G.~Popescu.
\newblock An algorithmic characterization of multi-dimensional mechanisms.
\newblock {\em Computing Reviews}, 2012.

\bibitem{puny2022frame}
O.~Puny, M.~Atzmon, H.~Ben-Hamu, E.~J. Smith, I.~Misra, A.~Grover, and
  Y.~Lipman.
\newblock Frame averaging for invariant and equivariant network design.
\newblock {\em {International Conference on Learning Representations (ICLR)}},
  2022.

\bibitem{rahme2021permutation}
J.~Rahme, S.~Jelassi, J.~Bruna, and S.~M. Weinberg.
\newblock A permutation-equivariant neural network architecture for auction
  design.
\newblock In {\em Proceedings of the AAAI Conference on Artificial
  Intelligence}, 2021.

\bibitem{2016Twenty}
T.~Roughgarden.
\newblock Twenty lectures on algorithmic game theory.
\newblock {\em Cambridge Books}, 2016.

\bibitem{sandholm2015automated}
T.~Sandholm and A.~Likhodedov.
\newblock Automated design of revenue-maximizing combinatorial auctions.
\newblock {\em Operations Research}, 2015.

\bibitem{2017A}
V.~Syrgkanis.
\newblock A sample complexity measure with applications to learning optimal
  auctions.
\newblock {\em {Neural Information Processing Systems (NeurIPS)}}, 2017.

\end{thebibliography}
\end{document}